\documentclass[11pt]{article}
\usepackage{amsfonts}
\usepackage{amsthm}
\usepackage{amsmath}
\usepackage{amssymb}
\usepackage{algorithm}
\usepackage{algorithmic}
\usepackage{mathtools}
\usepackage{hyperref}
\usepackage{xcolor}
\usepackage{listings}
\definecolor{codebg}{RGB}{247,247,250}
\definecolor{codeframe}{RGB}{200,200,210}
\definecolor{codekw}{RGB}{0,70,140}
\definecolor{codecomment}{RGB}{110,110,110}
\definecolor{codestring}{RGB}{140,60,0}
\lstdefinestyle{sage}{
  language=Python,
  basicstyle=\ttfamily\footnotesize,
  keywordstyle=\color{codekw}\bfseries,
  commentstyle=\color{codecomment}\itshape,
  stringstyle=\color{codestring},
  morekeywords={GF,VectorSpace,vector,Permutation,Permutations,identity_matrix,matrix,run},
  numbers=left, numberstyle=\tiny\color{codecomment}, numbersep=8pt,
  backgroundcolor=\color{codebg}, frame=single, rulecolor=\color{codeframe},
  framerule=0.4pt, xleftmargin=22pt, framexleftmargin=18pt,
  columns=fullflexible, keepspaces=true, showstringspaces=false,
  breaklines=true, tabsize=4, aboveskip=8pt, belowskip=8pt,
}

\newcommand{\FF}{\mathbb{F}}

\newcommand{\0}{\mathbf{0}}
\newcommand{\1}{\mathbf{1}}
\DeclareMathOperator{\Hull}{Hull}
\DeclareMathOperator{\rank}{rank}

\DeclareMathOperator{\disc}{disc}

\theoremstyle{plain}
\newtheorem{theorem}{Theorem}[section]
\newtheorem{lemma}[theorem]{Lemma}
\newtheorem{corollary}[theorem]{Corollary}
\newtheorem{proposition}[theorem]{Proposition}

\theoremstyle{definition}
\newtheorem{definition}[theorem]{Definition}
\newtheorem{example}[theorem]{Example}
\theoremstyle{remark}
\newtheorem{remark}[theorem]{Remark}

\title{The Closure of LCD-to-GI Reductions via Generalized Inner Products}

\author{Keita~Ishizuka
\thanks{K. Ishizuka is with the Information Technology R\&D Center, Mitsubishi Electric Corporation, Kanagawa, Japan (e-mail: keitaishizuka1994@gmail.com).}}
\date{}

\begin{document}

\maketitle

\begin{abstract}
The Permutation Equivalence Problem (PEP) for linear complementary dual (LCD) codes reduces to Graph Isomorphism (GI) through the orthogonal projector onto the code, but this requires a trivial hull. We ask whether the reduction extends when the standard inner product is replaced by a nondegenerate symmetric bilinear form~\(M\). We make precise what it means for \(M\) to yield a valid reduction to GI and prove that, for length \(n \geq 3\) and every field size, this holds if and only if \(M\) is permutation-invariant. The reduction then reaches, beyond LCD codes, a class of codes with one-dimensional hull: exactly those whose hull is spanned by a vector with nonzero coordinate sum. We count these codes exactly; for odd \(q\), asymptotically the method covers a fraction \(\prod_{i \geq 2}(1+q^{-i})^{-1}\) of all codes, \(1 + q^{-1}\) times the fraction of LCD codes. A polynomial-time reduction algorithm is given.
\end{abstract}


\section{Introduction}

The Permutation Equivalence Problem (PEP) for linear codes asks whether two given codes are identical up to a coordinate permutation. It underlies the classification of codes up to equivalence, and in code-based cryptography it is the hardness assumption behind the signature scheme SPECK~\cite{SPECK}; see~\cite{horlemann2026survey} for a survey. Two families of solvers are known. Sendrier's Support Splitting Algorithm (SSA)~\cite{sendrier2000finding} is the standard tool; it uses the hull $\Hull(C) = C \cap C^\perp$ as a signature, and its cost is dominated by a term exponential in $\dim\Hull(C)$, so codes with small hulls are its easy instances. The codeword-search solvers, from Leon's algorithm~\cite{leon1982} to Beullens's~\cite{beullens2020}, instead match low-weight codewords of the two codes, so their bottleneck is information-set decoding; through them, PEP is tied to the decoding problem as well. Canonical-form methods~\cite{feulner2009,bouyukliev2007} compute a distinguished representative of the equivalence class and are the tools behind the classification tables of codes; they apply to all codes but come with no worst-case guarantee.

On the complexity side, Petrank and Roth~\cite{petrank1997finding} showed that PEP is at least as hard as Graph Isomorphism (GI). Conversely, Bardet et al.~\cite{bardet2019permutation} proved that PEP for Linear Complementary Dual (LCD) codes, those with trivial hull, reduces to GI in polynomial time: for an LCD code with generator matrix~$G$, the orthogonal projector $\Pi_C = G^T(GG^T)^{-1}G$ is a complete invariant, and $\Pi_{CP} = P^T\Pi_CP$ for every permutation matrix~$P$, so two LCD codes are equivalent if and only if their projectors are isomorphic as edge-weighted graphs. With Babai's quasi-polynomial GI algorithm~\cite{babai2016}, and the practical strength of GI solvers such as nauty and Traces~\cite{mckay2014}, this transfers both a worst-case bound and effective software to PEP on LCD codes. The construction requires $GG^T$ to be nonsingular, which excludes every code with a nontrivial hull.

Since the hull dimension governs the cost of SSA, it is natural to ask whether it can be lowered. With this aim, hulls taken with respect to a bilinear form $xMy^T$ in place of the standard inner product were studied in~\cite{aikawa2026sac}. It was shown there that the $M$-hull can serve as an invariant inside SSA if and only if $M$ is permutation-invariant, that the permutation-invariant forms are exactly $M = aI + bJ$ with $J$ the all-ones matrix, and that for these forms the hull dimension drops by at most one. A code whose hull is one-dimensional may therefore have trivial $M$-hull, and one may hope to run the LCD-to-GI reduction with respect to~$M$ on codes that were out of reach. \emph{Does this work, for which $M$, and for which codes?} We answer these questions completely.

\paragraph{Contributions.}
In this paper we carry the generalized-inner-product approach of~\cite{aikawa2026sac} over to the LCD-to-GI reduction: we replace the standard inner product by a bilinear form~$M$ and ask whether the reduction then applies to codes that were previously out of reach. Our contributions are as follows.
(i) We make precise what it means for $M$ to \emph{yield a valid reduction} to GI (Definition~\ref{def:valid}).
(ii) We prove that, for $n \geq 3$ and every~$q$, $M$ yields a valid reduction if and only if $M$ is permutation-invariant (Theorem~\ref{thm:main_characterization}); with the known description of permutation-invariant forms, these are exactly $M = aI + bJ$ (Corollary~\ref{cor:aIbJ}).
(iii) We characterize the reducible codes: a code admits such a reduction if and only if its hull has dimension at most~$1$ and, when $\Hull(C) = \langle x\rangle$, the coordinate sum $(x,\1)$ is nonzero (Theorem~\ref{thm:gi_reducible_count}); in characteristic~$2$ only LCD codes are reducible (Corollary~\ref{cor:char2}).
(iv) We give a polynomial-time reduction algorithm with a correctness proof (Section~\ref{sec:algorithm}).
(v) We quantify the gain: the number of GI-reducible codes is given in closed form (Theorem~\ref{thm:enumeration_main}), and for odd~$q$, asymptotically the orthogonal projector method covers a fraction $(1 + q^{-1})\prod_{i \geq 1}(1 + q^{-i})^{-1}$ of all codes, exactly $1 + q^{-1}$ times the fraction $\prod_{i \geq 1}(1 + q^{-i})^{-1}$ of LCD codes reached by Bardet et al.\ (Proposition~\ref{prop:coverage}); for fixed dimension~$k$ the limits are given in Proposition~\ref{prop:fixedk}.
Together, these results give an LCD-to-GI reduction in the spirit of~\cite{bardet2019permutation} for a strictly larger class of codes over fields of odd characteristic, with an algorithm that is easy to implement. More importantly, the reduction is complete: within the framework of bilinear forms, no other reduction of this kind is possible.

\paragraph{Organization.}
Section~\ref{sec:preliminaries} collects the prior work: PEP, GI, the reduction of Bardet et al., bilinear and quadratic forms over finite fields, and the known results on $M$-hulls. Section~\ref{sec:extension} characterizes the forms that yield a valid reduction, Section~\ref{sec:characterization} the GI-reducible codes, and Section~\ref{sec:algorithm} presents the reduction algorithm. Section~\ref{sec:enumeration} counts the GI-reducible codes and Section~\ref{sec:coverage} derives the coverage results. Section~\ref{sec:conclusion} concludes with the consequences for the classification of codes and for the complexity of PEP. Appendix~\ref{app:lcd_enumeration} collects the LCD mass formulas, Appendix~\ref{app:type_epsilon} the type-$\epsilon$ version needed in Section~\ref{sec:enumeration}, and Appendix~\ref{app:script} the script used to validate the algorithm.

\section{Preliminaries}
\label{sec:preliminaries}

\subsection{Codes, graphs, and the LCD-to-GI reduction}
All vectors are row vectors over $\FF_q$, $q = p^m$; $v^T$ is the column vector, $(x,y) = \sum_i x_iy_i$ is the standard inner product, $\1 = (1,\ldots,1)$, $J = \1^T\1$, and $\langle\cdot\rangle$ denotes a span. For $\pi \in S_n$ the permutation matrix $P_\pi$ has $(P_\pi)_{ij} = 1$ if $j = \pi(i)$, so that $\pi$ acts on row vectors by $x \mapsto xP_\pi$. An $[n,k]_q$ code is a $k$-dimensional subspace $C \subseteq \FF_q^n$ with generator matrix $G \in \FF_q^{k\times n}$, dual $C^\perp = \{y : (x,y) = 0\ \forall x \in C\}$ and hull $\Hull(C) = C \cap C^\perp$; $C$ is \emph{LCD} if $\Hull(C) = \{0\}$. For a permutation matrix~$P$, $CP = \{cP : c \in C\}$, and $C_1 \cong C_2$ means $C_2 = C_1P$ for some~$P$.

\begin{definition}[Permutation Equivalence Problem]
Given $[n,k]_q$ codes $C_1, C_2$, decide whether $C_1 \cong C_2$.
\end{definition}

\begin{definition}[Graph Isomorphism]
Matrices $A_1, A_2 \in \FF_q^{n\times n}$ are read as weighted adjacency matrices of directed graphs with loops; they are isomorphic, $A_1 \cong A_2$, if $A_2 = P^TA_1P$ for some permutation matrix~$P$. GI is the problem of deciding $A_1 \cong A_2$.
\end{definition}

Deciding $A_1 \cong A_2$ reduces to isomorphism of unweighted undirected graphs in polynomial time by standard gadgets~\cite{bardet2019permutation}; for weights in $\FF_q$ the resulting unweighted graph has $O(n^2\log q)$ vertices. Babai~\cite{babai2016} showed that GI on $N$ vertices can be solved in time $2^{O(\log^cN)}$ for some constant~$c$. PEP is at least as hard as GI~\cite{petrank1997finding}. In the other direction:

\begin{theorem}[Massey~\cite{massey1992linear}]
\label{thm:massey}
$C$ is LCD if and only if $GG^T$ is nonsingular, and then $\Pi_C = G^T(GG^T)^{-1}G$ is the orthogonal projector onto~$C$.
\end{theorem}

\begin{theorem}[Bardet et al.~\cite{bardet2019permutation}]
\label{thm:gi_reduction}
For LCD codes, $\Pi_C$ depends only on~$C$ and $\Pi_{CP} = P^T\Pi_CP$; hence $C_1 \cong C_2$ if and only if $\Pi_{C_1} \cong \Pi_{C_2}$, and PEP on LCD codes reduces to GI in polynomial time.
\end{theorem}

At this point two small examples may be helpful.

\begin{example}
\label{ex:lcd}
Over $\FF_3$, let $C = \langle(1,1,0,0),(0,0,1,1)\rangle$ with generator matrix~$G$, and let $C' = CP$ for the transposition $\sigma = (2\;3)$ with permutation matrix~$P$. Then $GG^T = \mathrm{diag}(2,2)$ is nonsingular, so $C$ is LCD, and
\[
\Pi_C = \begin{psmallmatrix} 2&2&0&0\\2&2&0&0\\0&0&2&2\\0&0&2&2\end{psmallmatrix}, \qquad
\Pi_{C'} = \begin{psmallmatrix} 2&0&2&0\\0&2&0&2\\2&0&2&0\\0&2&0&2\end{psmallmatrix} = P^T\Pi_CP .
\]
Thus the equivalence of $C$ and $C'$ is witnessed by an isomorphism of the two weighted graphs, as in~\cite{bardet2019permutation}.
\end{example}

\begin{example}
\label{ex:intro}
Over $\FF_3$, let $C = \langle(1,1,0,0),(0,1,1,0)\rangle$ with generator matrix~$G$, and let $C' = CP$ for the transposition $\sigma = (1\;2)$ with permutation matrix~$P$. Now $GG^T = \begin{psmallmatrix}2&1\\1&2\end{psmallmatrix}$ is singular, so $C$ (and hence $C'$) is not LCD, the orthogonal projectors are undefined, and the reduction of~\cite{bardet2019permutation} does not apply to $C$ and $C'$.
\end{example}

In this paper we show that the code of Example~\ref{ex:intro} does reduce to GI, and how: it is taken up again in Example~\ref{ex:gi_reduction}, after the theory has been developed.

\subsection{Quadratic and Bilinear Forms}

We summarize the structural facts on quadratic and symmetric bilinear forms used throughout the paper; the standard references are~\cite{lidl1997finitefields} and~\cite{taylor1992geometry}. Throughout this subsection $q$ is odd unless stated otherwise.

A \emph{bilinear form} on $\FF_q^n$ is a map $B: \FF_q^n \times \FF_q^n \rightarrow \FF_q$ that is linear in each argument. Every bilinear form $B$ can be represented by a matrix $M \in \FF_q^{n \times n}$ such that $B(x, y) = xMy^T$ for all $x, y \in \FF_q^n$, where we view $x$ and $y$ as row vectors. We call $M$ the \emph{structure matrix} of $B$. A bilinear form is \emph{nondegenerate} if $\det(M) \neq 0$. Throughout this paper, we restrict to \emph{symmetric} bilinear forms (i.e., $M = M^T$): a non-symmetric structure matrix would necessitate distinguishing left and right dual codes, which falls outside our framework. We identify a bilinear form with its structure matrix with respect to the standard basis, and use the matrix notation $M$ directly instead of the function notation $B$.
For a subspace $W \subseteq \FF_q^n$ we write
\[
\Hull_M(W) := \{v \in W : B(v, w) = 0 \text{ for all } w \in W\}
\]
for the $M$-hull of $W$ (Definition~\ref{def:Mhull} below). $W$ is \emph{$M$-LCD} if $\Hull_M(W) = \{0\}$; equivalently, the Gram matrix of any basis of~$W$ is nonsingular (see Remark~\ref{rem:terminology}). A vector $x$ is \emph{isotropic} if  $B(x,x) = 0$.
The \emph{discriminant} of a nondegenerate $B$ on $\FF_q^n$ is the class $\disc(B) := \det(M) \cdot (\FF_q^*)^2 \in \FF_q^*/(\FF_q^*)^2$. This is invariant under change of basis $M \mapsto AMA^T$ ($A \in \mathrm{GL}_n$, $\det(AMA^T) = \det(A)^2 \det(M)$). Two nondegenerate symmetric bilinear forms on $\FF_q^n$ ($q$ odd) are \emph{isometric} (related by such a change of basis) if and only if they have the same discriminant. We say that a nondegenerate form (or its structure matrix) has \emph{type} $\epsilon = \chi(\det M) \in \{+1,-1\}$, where $\chi$ is the quadratic character defined below; type $+1$ is the class of the standard inner product, type $-1$ that of $\mathrm{diag}(1,\ldots,1,\gamma)$ for a non-square~$\gamma$.
A two-dimensional nondegenerate $B$ with isotropic basis $\{u, v\}$ satisfying $B(u,v) = 1$ is called a \emph{hyperbolic plane}; its Gram matrix is $\begin{psmallmatrix} 0 & 1 \\ 1 & 0 \end{psmallmatrix}$, so $\disc = -1$. A subspace on which $B$ vanishes only at the origin is called \emph{anisotropic}.

\begin{remark}[Terminology]
\label{rem:terminology}
For a subspace $W$ of $\FF_q^n$, the \emph{hull} $\Hull_M(W)$ is the \emph{radical} of the restricted form $B|_W$, and $W$ is \emph{$M$-LCD} if and only if $B|_W$ is \emph{nondegenerate}; the two terminologies refer to the same objects. We use the coding-theoretic terms when $W$ is treated as a code, and the form-theoretic terms when $W$ is treated as a space carrying the restricted form, as in quotients by the radical, discriminants and counting arguments.
\end{remark}

The following three structural results from the theory of bilinear forms over $\FF_q$ are used repeatedly; all are standard~\cite{taylor1992geometry}.

\begin{lemma}[Hyperbolic completion]\label{lem:hyperbolic_completion}
Let $q$ be odd and let $B$ be a nondegenerate symmetric bilinear form on $V = \FF_q^n$. For every nonzero isotropic vector $u \in V$, there exists $v \in V$ such that $\{u, v\}$ spans a hyperbolic plane.
\end{lemma}

\begin{lemma}[Existence of an orthogonal basis]\label{lem:orth_basis}
Let $q$ be odd. Every nondegenerate symmetric bilinear form $B$ on $\FF_q^n$ admits an orthogonal basis: the structure matrix $M$ is congruent to $\mathrm{diag}(d_1, \ldots, d_n)$ for some $d_1, \ldots, d_n \in \FF_q^*$.
\end{lemma}

\begin{theorem}[Witt decomposition]\label{thm:witt_decomp}
Let $q$ be odd and let $B$ be a nondegenerate symmetric bilinear form on $V = \FF_q^n$. Then $V$ decomposes as
\[
V = H \perp W,
\]
where $H$ is an orthogonal sum of hyperbolic planes (the \emph{hyperbolic part}) and $W$ is anisotropic with $\dim W \leq 2$. The decomposition is unique up to isometry.
\end{theorem}

A \emph{quadratic form} on $\FF_q^n$ is a homogeneous polynomial of degree~$2$. For a symmetric structure matrix~$M$ we write $Q_M(x) := xMx^T$ for the associated quadratic form and $B_M(x,y) := xMy^T$ for the bilinear form, so that
\[
Q_M(x+y) = Q_M(x) + Q_M(y) + 2B_M(x,y), \qquad B_M(x,x) = Q_M(x).
\]
For $q$ odd, $M \mapsto Q_M$ is a bijection between symmetric matrices and quadratic forms (polarization), and we identify a quadratic form, its bilinear form, and its symmetric matrix~$M$; all vectors are row vectors, and $v^T$ denotes the corresponding column vector.
When $q$ is odd, the \emph{quadratic character} $\chi: \FF_q^* \to \{-1, +1\}$ is defined by $\chi(a) = 1$ if $a \in (\FF_q^*)^2$ and $\chi(a) = -1$ otherwise. We extend $\chi$ to $\FF_q$ by setting $\chi(0) = 0$. The discriminant class of $B$ is encoded by $\chi(\det M) \in \{+1, -1\}$. The following classical result counts solutions to quadratic forms.

\begin{theorem}[Weil; Lidl and Niederreiter {\cite{lidl1997finitefields}}]\label{thm:weil_lidl_niederreiter}
Let $Q(x) = \sum_{i,j} Q_{ij} x_i x_j$ be a nondegenerate quadratic form over $\mathbb{F}_q^n$ where $q$ is odd, and let $N(Q) = |\{x \in \mathbb{F}_q^n : Q(x) = 0\}|$ be the number of its solutions. Then
$$
N(Q) = \begin{cases}
    q^{n-1} & \textrm{ if $n$ is odd,} \\
    q^{n-1} + \chi((-1)^{n/2} \det(Q))(q-1)q^{(n-2)/2} & \textrm{ if $n$ is even,}
\end{cases}
$$
where $\chi$ is the quadratic character of $\mathbb{F}_q$ and $\det(Q)$ denotes the determinant of the matrix of $Q$.
\end{theorem}

\subsection{Generalized hulls}
Replacing the standard inner product by a general bilinear form changes the notions of dual and hull, and with them the class of ``LCD'' codes.

\begin{definition}[$M$-hull~\cite{aikawa2026sac}]
\label{def:Mhull}
For a code $C$ and a nondegenerate symmetric structure matrix~$M$,
\[
C^{\perp_M} := \{y \in \FF_q^n : xMy^T = 0 \text{ for all } x \in C\}, \qquad
\Hull_M(C) := C \cap C^{\perp_M},
\]
and $C$ is \emph{$M$-LCD} if $\Hull_M(C) = \{0\}$.
\end{definition}

$M = I$ recovers the usual notions. Since $c \mapsto cG$ is injective on $\FF_q^k$, $\Hull_M(C) = \{cG : cGMG^T = 0\}$ and
\begin{equation}
\label{eq:hull_dim}
\dim\Hull_M(C) = k - \rank(GMG^T);
\end{equation}
in particular $C$ is $M$-LCD if and only if $GMG^T$ is nonsingular. We say that $M$ is \emph{permutation-invariant} if $PMP^T = M$ for all permutation matrices~$P$. The following facts about $M$-hulls are known.

\begin{theorem}[{\cite[Theorem~1]{aikawa2026sac}}]
\label{thm:sac}
Let $M$ be nondegenerate symmetric and $n \geq 3$. The following are equivalent:
\begin{enumerate}
\item[(a)] $\Hull_M(CP) = \Hull_M(C)P$ for all codes~$C$ and all permutation matrices~$P$, i.e.\ the $M$-hull is compatible with the Support Splitting Algorithm;
\item[(b)] $M$ is permutation-invariant;
\item[(c)] $M = aI + bJ$ for some $a, b \in \FF_q$.
\end{enumerate}
\end{theorem}

\begin{remark}[Non-degeneracy of $M = aI + bJ$]
\label{rem:nondegeneracy}
A direct computation shows
\[
\det(aI + bJ) = a^{n-1}(a + nb),
\]
where $n$ is treated as an element of $\mathbb{F}_q$. Thus $M = aI + bJ$ is nondegenerate if and only if $a \neq 0$ and $a + nb \neq 0$. (When $p \mid n$, the second condition reduces to $a \neq 0$.) Throughout this paper, when we write $M = aI + bJ$, we implicitly assume these conditions hold.
\end{remark}

\begin{theorem}[{\cite[Theorem~2]{aikawa2026sac}}]
\label{thm:dimension_bound}
For nondegenerate $M = aI + bJ$ and every code~$C$,
\[
|\dim\Hull_M(C) - \dim\Hull(C)| \leq 1.
\]
\end{theorem}

Theorem~\ref{thm:sac} says which $M$-hulls can serve as invariants inside SSA, and Theorem~\ref{thm:dimension_bound} says what is gained. The questions of this paper are different: for which~$M$ does the LCD-to-GI reduction of Theorem~\ref{thm:gi_reduction} go through with $\Hull_M$ in place of $\Hull$, and which codes does it then reach?

\section{Which Forms Yield a Valid Reduction}
\label{sec:extension}
The first ingredient is the invariant that replaces the orthogonal projector.

\begin{definition}[$M$-orthogonal projector]
\label{def:proj}
For a nondegenerate symmetric~$M$ and an $M$-LCD code $C$ with generator matrix~$G$, put $\Pi_{C,M} := MG^T(GMG^T)^{-1}G$.
\end{definition}

\begin{theorem}
\label{thm:generalized_massey}
$\Pi_{C,M}$ is idempotent with row space~$C$, satisfies $vMw^T = 0$ whenever $v\Pi_{C,M} = v$ and $w\Pi_{C,M} = 0$, and depends only on~$C$, not on~$G$. For $M = I$ it is the orthogonal projector of Theorem~\ref{thm:massey}.
\end{theorem}
\begin{proof}
$G\Pi_{C,M} = G$ gives $\Pi_{C,M}^2 = \Pi_{C,M}$ and row space $\supseteq C$; every row of $\Pi_{C,M}$ is a combination of rows of~$G$, so the row space is~$C$. If $w\Pi_{C,M} = 0$ then $wMG^T(GMG^T)^{-1}G = 0$, and since $(GMG^T)^{-1}G$ has full row rank, $wMG^T = 0$; hence $vMw^T = cGMw^T = 0$ for $v = cG$. Replacing $G$ by $AG$ with $A \in \mathrm{GL}_k$ leaves $\Pi_{C,M}$ unchanged.
\end{proof}

$\Pi_{C,M}$ is read as the weighted adjacency matrix of a directed graph on $n$ vertices; for $M \neq I$ it is in general not symmetric. The second ingredient is a guarantee that $C \mapsto \Pi_{C,M}$ is a reduction of PEP to GI. This is not automatic, and the main result of this section determines exactly for which~$M$ it holds.

\begin{definition}[Valid reduction]
\label{def:valid}
A nondegenerate symmetric $M$ \emph{yields a valid reduction} if
\begin{enumerate}
\item[(V1)] the family of $M$-LCD codes is closed under coordinate permutations: if $C$ is $M$-LCD then so is $CP$ for every permutation matrix~$P$; and
\item[(V2)] for all $M$-LCD codes $C_1, C_2 \subseteq \FF_q^n$ of the same dimension, $C_1 \cong C_2 \iff \Pi_{C_1,M} \cong \Pi_{C_2,M}$.
\end{enumerate}
\end{definition}

In the remainder of this section we identify the bilinear forms that yield a valid reduction. They coincide with the forms compatible with the Support Splitting Algorithm: $M$ yields a valid reduction exactly when it is permutation-invariant. We begin with three auxiliary lemmas.

\begin{lemma}[Orbit scalars, {\cite[Lemma~2]{aikawa2026sac}}]
\label{lem:orbit}
Let $n \geq 3$ and let $N \neq 0$ be a symmetric matrix over $\FF_q$ such that, for every permutation matrix $P$, $PNP^T = \nu_P N$ for some $\nu_P \in \FF_q^*$. Then $\nu_P = 1$ for all $P$, i.e., $PNP^T = N$.
\end{lemma}
\begin{proof}
Since $N \neq 0$, each $\nu_P$ is uniquely determined, and $(P_1P_2)N(P_1P_2)^T = P_1(P_2NP_2^T)P_1^T$ shows that $P \mapsto \nu_P$ is a group homomorphism $S_n \to \FF_q^*$. Its target is abelian, so $\nu$ is constant on conjugacy classes; in particular all transpositions have the same value $\nu_\tau$, and $\nu_\tau^2 = 1$. If $q$ is even, $|\FF_q^*| = q-1$ is odd, so $\nu_\tau = 1$. Suppose $q$ is odd and $\nu_\tau = -1$. For $k \neq l$, the transposition $\tau = (k\;l)$ gives $N_{kl} = N_{\tau(l)\tau(k)} = (\tau N \tau^T)_{lk} = -N_{lk} = -N_{kl}$, so $N_{kl} = 0$. For a diagonal entry $N_{kk}$, choose a transposition $\tau = (a\;b)$ with $k \notin \{a,b\}$, which exists since $n \geq 3$; then $N_{kk} = (\tau N\tau^T)_{kk} = -N_{kk}$, so $N_{kk} = 0$. Hence $N = 0$, a contradiction. Thus $\nu_\tau = 1$ for all transpositions, which generate $S_n$, and $\nu_P = 1$ for all~$P$.
\end{proof}

\begin{lemma}[Odd characteristic]
\label{lem:quadratic_proportionality}
Let $q$ be odd, $n \geq 3$, and let $M$ be nondegenerate symmetric over $\FF_q$. Suppose that for every permutation matrix~$P$, every code of dimension one is $M$-LCD if and only if it is $PMP^T$-LCD. Then $PMP^T = M$ for every permutation matrix~$P$.
\end{lemma}
\begin{proof}
Fix~$P$ and put $M' := PMP^T$, which is again symmetric and nondegenerate. For one-dimensional codes the hypothesis says that $xMx^T = 0 \iff xM'x^T = 0$ for all $x \in \FF_q^n$, so the quadratic forms $Q_1 := Q_M$ and $Q_2 := Q_{M'}$ have the same zero set. Let $B_1, B_2$ be their bilinear forms, and call a vector \emph{anisotropic} if $Q_1(u) \neq 0$ (equivalently $Q_2(u) \neq 0$). We use two facts about a nondegenerate quadratic form $Q$ in $n \geq 3$ variables over $\FF_q$, $q$ odd.

\emph{Fact 1: the isotropic vectors span $\FF_q^n$.} By the Witt decomposition (Theorem~\ref{thm:witt_decomp}), $\FF_q^n = H \perp W$ with $H \neq 0$ an orthogonal sum of hyperbolic planes (every form in at least three variables over a finite field is isotropic) and $W$ anisotropic of dimension at most~$2$. The hyperbolic bases of $H$ are isotropic and span~$H$. A hyperbolic plane represents every element of $\FF_q$: for an isotropic basis $a, b$ with $B(a,b) = 1$, $Q(a + tb) = 2tB(a,b) = 2t$ runs through $\FF_q$ as $t$ does. Hence for $w \in W \setminus \{0\}$ there is $h \in H$ with $Q(h) = -Q(w)$, and $h + w$ is isotropic, since $Q(h + w) = Q(h) + Q(w) + 2B(h,w) = Q(h) + Q(w) = 0$. These vectors span $\FF_q^n$.

\emph{Fact 2: the anisotropic vectors span $\FF_q^n$.} An orthogonal basis $v_1, \ldots, v_n$ (Lemma~\ref{lem:orth_basis}) has Gram matrix $\mathrm{diag}(Q(v_1), \ldots, Q(v_n))$, which is nonsingular; so every $v_i$ is anisotropic, and the anisotropic vectors contain a basis.

\emph{Step 1.} Let $u$ be anisotropic and put $\lambda_u := Q_2(u)/Q_1(u) \in \FF_q^*$. For an isotropic vector $w$ consider $f_i(t) := Q_i(u + tw) = Q_i(u) + 2tB_i(u,w)$, $i = 1,2$, which are polynomials of degree at most~$1$ in~$t$ with the same zero set in~$\FF_q$ (the zero sets of $Q_1$ and $Q_2$ coincide). If $B_1(u,w) = 0$ then $f_1$ has no zero, so $B_2(u,w) = 0$; otherwise $f_1$ has the unique zero $-Q_1(u)/(2B_1(u,w))$, which must equal $-Q_2(u)/(2B_2(u,w))$, giving $B_2(u,w) = \lambda_u B_1(u,w)$. Hence the linear functional $B_2(u,\cdot) - \lambda_u B_1(u,\cdot)$ vanishes on all isotropic vectors, and by Fact~1 it vanishes identically:
\begin{equation}
\label{eq:prop_step1}
B_2(u, \cdot) = \lambda_u B_1(u, \cdot) \qquad \text{for every anisotropic } u.
\end{equation}

\emph{Step 2: $\lambda_u$ does not depend on~$u$.} If $u, u'$ are anisotropic with $B_1(u,u') \neq 0$, then by~\eqref{eq:prop_step1} and symmetry,
\[
\lambda_u B_1(u,u') = B_2(u,u') = B_2(u',u) = \lambda_{u'} B_1(u',u) = \lambda_{u'} B_1(u,u'),
\]
so $\lambda_u = \lambda_{u'}$. If $B_1(u,u') = 0$, it suffices to exhibit an anisotropic $v$ with $B_1(u,v) \neq 0 \neq B_1(u',v)$, since then $\lambda_u = \lambda_v = \lambda_{u'}$. Put $v := u + u'$; then $B_1(u,v) = Q_1(u) \neq 0$ and $B_1(u',v) = Q_1(u') \neq 0$, and $Q_1(v) = Q_1(u) + Q_1(u')$. If this is nonzero we are done. Otherwise $Q_1(u') = -Q_1(u)$; for $q \geq 5$ replace $u'$ by $cu'$ with $c^2 \neq 1$ (note $\lambda_{cu'} = \lambda_{u'}$), so that $Q_1(u + cu') = (1 - c^2)Q_1(u) \neq 0$. For $q = 3$ every $c \in \FF_q^*$ has $c^2 = 1$, and we use a third vector: the restriction of $B_1$ to $\langle u,u'\rangle$ is nondegenerate (its Gram matrix is $\mathrm{diag}(Q_1(u),Q_1(u'))$), so its restriction to $\langle u,u'\rangle^\perp$ is nondegenerate as well, of dimension $n - 2 \geq 1$, and $\langle u,u'\rangle^\perp$ contains an anisotropic~$z$; then $v := u + u' + z$ has $B_1(u,v) = Q_1(u)$, $B_1(u',v) = Q_1(u')$ and $Q_1(v) = Q_1(z) \neq 0$. In all cases $\lambda_u = \lambda_{u'}$; write $\lambda$ for the common value.

\emph{Step 3.} By~\eqref{eq:prop_step1}, $B_2(u,\cdot) = \lambda B_1(u,\cdot)$ for every anisotropic~$u$, and by Fact~2 these span $\FF_q^n$. Hence $B_2 = \lambda B_1$, that is, $M' = \lambda_P M$ with $\lambda_P := \lambda \in \FF_q^*$. As $P$ was arbitrary, Lemma~\ref{lem:orbit} (with $N = M \neq 0$) gives $\lambda_P = 1$, hence $PMP^T = M$.
\end{proof}

\begin{lemma}[Even characteristic]
\label{lem:even}
Let $q$ be even, $n \geq 3$, and let $M$ be nondegenerate symmetric over $\FF_q$. Suppose that for every permutation matrix~$P$, every code of dimension at most two is $M$-LCD if and only if it is $PMP^T$-LCD. Then $PMP^T = M$ for every permutation matrix~$P$.
\end{lemma}
\begin{proof}
Fix~$P$ and put $M' := PMP^T$, which is again symmetric and nondegenerate. For one-dimensional codes the hypothesis says that $xMx^T = 0 \iff xM'x^T = 0$ for all $x \in \FF_q^n$. Now $xMx^T = \sum_i M_{ii}x_i^2 = (d\,x^T)^2$, where $d \in \FF_q^n$ is the vector with entries $d_i := \sqrt{M_{ii}}$; so the common zero set $Z = \{x : d\,x^T = 0\}$ is $\FF_q^n$ if $d = 0$ and a hyperplane otherwise, and in both cases $\dim Z \geq n - 1 \geq 2$. Let $x \in Z \setminus \{0\}$ and $y \in \FF_q^n$. If $y \notin \langle x \rangle$, the Gram matrix of $\langle x, y\rangle$ with respect to $M$ is $\begin{psmallmatrix} 0 & xMy^T \\ yMx^T & yMy^T\end{psmallmatrix}$, whose determinant is $-(xMy^T)^2 = (xMy^T)^2$; so the two-dimensional code $\langle x,y\rangle$ is $M$-LCD iff $xMy^T \neq 0$, and likewise it is $M'$-LCD iff $xM'y^T \neq 0$. By the hypothesis for two-dimensional codes, $xMy^T = 0 \iff xM'y^T = 0$; if $y \in \langle x\rangle$ the same holds by the one-dimensional case. Hence the nonzero linear functionals $xM$ and $xM'$ have the same kernel, so $xM' = \mu(x)\,xM$ for some $\mu(x) \in \FF_q^*$.
We claim that $\mu$ is constant on $Z \setminus \{0\}$. Clearly $\mu(cx) = \mu(x)$ for $c \in \FF_q^*$. If $x, y \in Z$ are linearly independent then $x + y \in Z$ and
\[
\mu(x+y)(xM + yM) = (x+y)M' = \mu(x)\,xM + \mu(y)\,yM;
\]
since $xM, yM$ are linearly independent ($M$ is invertible), $\mu(x+y) = \mu(x) = \mu(y)$. Any two nonzero vectors of $Z$ are proportional or linearly independent, so $\mu$ is a constant $\lambda_P \in \FF_q^*$, and $x(M' - \lambda_P M) = 0$ for all $x \in Z$.

If $d = 0$, put $d_0 := 0$; then $Z = \FF_q^n$ and $M' = \lambda_P M$. If $d \neq 0$, then $M' - \lambda_P M$ is symmetric with left kernel containing the hyperplane~$Z$, so its row space is contained in $\langle d \rangle$ and $M' - \lambda_P M = c_P\, d^Td$ for some $c_P \in \FF_q$; that is, $PMP^T = \lambda_P M + c_P\, d^T d$. Comparing diagonal entries and using $d_i^2 = M_{ii}$ gives $d_{\pi(i)}^2 = (\lambda_P + c_P)\,d_i^2$ for all~$i$, where $\pi$ is the permutation of~$P$. If some $d_i = 0$ then $d_{\pi(i)} = 0$ for every~$\pi$, forcing $d = 0$; so every $d_i \neq 0$. For a transposition $\tau = (a\;b)$ and an index $c \notin \{a,b\}$ (which exists as $n \geq 3$), the entry $(c,c)$ gives $\lambda_\tau + c_\tau = 1$, whence $d_{\tau(i)}^2 = d_i^2$ and $d_a = d_b$; ranging over transpositions, $d = d_0\,\1$ with $d_0 \neq 0$ and $d^Td = d_0^2 J$. The diagonal comparison for an arbitrary~$P$ now reads $d_0^2 = (\lambda_P + c_P)d_0^2$, so $c_P = 1 + \lambda_P$ in characteristic~$2$, and $PMP^T = \lambda_PM + (1 + \lambda_P)d_0^2J$.

In both cases, using $PJP^T = J$ and $1 + \lambda_P + 1 = \lambda_P$ in characteristic~$2$,
\[
P(M + d_0^2J)P^T = \lambda_P\,(M + d_0^2J),
\]
and $M + d_0^2J \neq 0$ since $M$ is nondegenerate and $d_0^2J$ has rank at most $1 < n$. Lemma~\ref{lem:orbit} gives $\lambda_P = 1$, hence $PMP^T = M$.
\end{proof}

\begin{theorem}[Characterization of valid reductions]
\label{thm:main_characterization}
Let $M$ be nondegenerate symmetric over $\FF_q$ and $n \geq 3$. Then $M$ yields a valid reduction if and only if $M$ is permutation-invariant.
\end{theorem}
\begin{proof}
Throughout, $P$ denotes a permutation matrix and $M' := PMP^T$, which is again symmetric and nondegenerate. For a code $C$ with generator matrix~$G$, the code $CP$ has generator matrix $GP$ and $(GP)M(GP)^T = GM'G^T$; hence
\begin{equation}
\label{eq:CP_Mprime}
CP \text{ is } M\text{-LCD} \iff C \text{ is } M'\text{-LCD}.
\end{equation}

($\Leftarrow$) (V1) follows from~\eqref{eq:CP_Mprime} since $M' = M$. For~(V2), let $C_1, C_2$ be $M$-LCD with generator matrices $G_1, G_2$. If $C_2 = C_1P$, we may take $G_2 = G_1P$, and then $\Pi_{C_2,M} = MP^TG_1^T(G_1MG_1^T)^{-1}G_1P = P^T\Pi_{C_1,M}P$, using $PMP^T = M$ and $MP^T = P^TM$. Conversely, if $\Pi_{C_2,M} = P^T\Pi_{C_1,M}P$, then taking row spaces (Theorem~\ref{thm:generalized_massey}) gives $C_2 = C_1P$. Hence $C_1 \cong C_2$.

($\Rightarrow$) Fix~$P$. Applying~(V1) to $P$ and to $P^{-1}$ and using~\eqref{eq:CP_Mprime}, a code $C$ is $M$-LCD if and only if it is $M'$-LCD. As $P$ was arbitrary, the hypotheses of Lemmas~\ref{lem:quadratic_proportionality} and~\ref{lem:even} hold.

\emph{Case $q$ odd.} Lemma~\ref{lem:quadratic_proportionality} gives $PMP^T = M$ for all~$P$.

\emph{Case $q$ even.} Here $Q_M$ depends only on the diagonal of~$M$, and one-dimensional codes do not suffice; Lemma~\ref{lem:even}, which uses codes of dimension at most two, gives $PMP^T = M$ for all~$P$.
\end{proof}

Combining Theorem~\ref{thm:main_characterization} with the known description of the permutation-invariant forms (Theorem~\ref{thm:sac}) gives the explicit answer.

\begin{corollary}
\label{cor:aIbJ}
Let $n \geq 3$. A nondegenerate symmetric $M$ yields a valid reduction if and only if $M = aI + bJ$ with $a \neq 0$ and $a + nb \neq 0$.
\end{corollary}

Corollary~\ref{cor:aIbJ} is a completeness statement: the generalization of the LCD-to-GI reduction to symmetric bilinear forms is exhausted by the two-parameter family $aI + bJ$, and no other symmetric bilinear form yields a reduction of this kind. Theorem~\ref{thm:sac} and Corollary~\ref{cor:aIbJ} together say that the forms usable inside SSA and the forms usable in the LCD-to-GI reduction are the same, although the two requirements are different.

\section{Which Codes Are Reducible}
\label{sec:characterization}
Having established which bilinear forms~$M$ yield a valid reduction, we now identify the class of codes whose permutation equivalence is thereby reduced to GI.

\begin{definition}[GI-reducible code]
\label{def:gi_reducible}
Let $n \geq 3$. A code $C \subseteq \FF_q^n$ is \emph{GI-reducible} (via the orthogonal projector method) if $C$ is $M$-LCD for some nondegenerate $M = aI + bJ$.
\end{definition}

By the standard theory of bilinear forms~\cite{taylor1992geometry}, every code decomposes as $C = \Hull(C) \perp U$ with $U$ an LCD code. We also use the Schur complement formula.

\begin{lemma}[Schur complement]
\label{lem:schur_complement}
Let $A \in \FF_q^{r\times r}$ be invertible, $B \in \FF_q^{r\times m}$, $E \in \FF_q^{m\times r}$ and $D \in \FF_q^{m\times m}$. Then $\det\begin{psmallmatrix} A & B \\ E & D\end{psmallmatrix} = \det(A)\det(D - EA^{-1}B)$.
\end{lemma}
\begin{proof}
$\begin{psmallmatrix} I & 0 \\ -EA^{-1} & I\end{psmallmatrix}\begin{psmallmatrix} A & B \\ E & D\end{psmallmatrix} = \begin{psmallmatrix} A & B \\ 0 & D - EA^{-1}B\end{psmallmatrix}$; take determinants.
\end{proof}

\begin{theorem}[Complete Characterization]
\label{thm:gi_reducible_count}
Let $n \geq 3$. An $[n,k]_q$ code $C$ is GI-reducible if and only if $C$ is LCD, or $\Hull(C) = \langle x\rangle$ is one-dimensional with $(x,\1) \neq 0$.
\end{theorem}
\begin{proof}
If $\dim\Hull(C) \geq 2$, Theorem~\ref{thm:dimension_bound} gives $\dim\Hull_M(C) \geq 1$ for every $M = aI + bJ$, so $C$ is not GI-reducible. LCD codes are GI-reducible with $M = I$. Let $\Hull(C) = \langle x\rangle$ and $C = \langle x\rangle \perp U$ with $U$ LCD, generated by~$G'$; put $G = \begin{psmallmatrix} x\\ G'\end{psmallmatrix}$ and $w = G'\1^T$. Since $(x,x) = 0$ and $xG'^T = 0$, for $M = aI + bJ$,
\[
GMG^T = \begin{pmatrix} b(x,\1)^2 & b(x,\1)w^T \\ b(x,\1)w & aG'G'^T + bww^T\end{pmatrix}.
\]
If $(x,\1) = 0$, the first row vanishes for all $a, b$, so $C$ is not $M$-LCD for any $M = aI + bJ$. If $(x,\1) \neq 0$, take $b \neq 0$; Lemma~\ref{lem:schur_complement} with $A = b(x,\1)^2$ gives
\begin{align*}
\det(GMG^T) &= b(x,\1)^2\det(aG'G'^T + bww^T - bww^T) \\
&= b(x,\1)^2\,a^{k-1}\det(G'G'^T) \neq 0,
\end{align*}
since $a \neq 0$ and $U$ is LCD. Such a pair exists whenever $(x,\1) \neq 0$ is possible at all: this forces $q$ odd (see Corollary~\ref{cor:char2}), and then $a = 1$ together with any of the at least $q - 2 \geq 1$ elements $b \in \FF_q^*$ with $1 + nb \neq 0$ yields a nondegenerate~$M$.
\end{proof}

\begin{corollary}\label{cor:char2}
If $q = 2^m$, then $C$ is GI-reducible if and only if $C$ is LCD.
\end{corollary}
\begin{proof}
In characteristic $2$, $(x,x) = \sum x_i^2 = (\sum x_i)^2 = (x,\1)^2$, so $x \in \Hull(C)$ implies $(x,x) = 0$ and hence $(x,\1) = 0$. By Theorem~\ref{thm:gi_reducible_count}, no code with $\dim(\Hull(C)) = 1$ is GI-reducible.
\end{proof}

\begin{example}
\label{ex:gi_reduction}
We return to the code $C = \langle(1,1,0,0),(0,1,1,0)\rangle \subseteq \FF_3^4$ of Example~\ref{ex:intro}, with generator matrix~$G$. The left kernel of $GG^T = \begin{psmallmatrix}2&1\\1&2\end{psmallmatrix}$ is $\langle(1,1)\rangle$, so $\Hull(C) = \langle x\rangle$ with $x = (1,2,1,0)$, and $(x,\1) = 1 \neq 0$: by Theorem~\ref{thm:gi_reducible_count}, $C$ is GI-reducible. Indeed, with $M = I + J$, $GMG^T = \begin{psmallmatrix}0&2\\2&0\end{psmallmatrix}$ is invertible. Let $\sigma = (1\;2) \in S_4$ with permutation matrix~$P$ and $C' = CP$; then
\[
\Pi_{C,M} = \begin{psmallmatrix} 1&1&0&0\\0&0&0&0\\0&1&1&0\\1&2&1&0\end{psmallmatrix}, \qquad
\Pi_{C',M} = \begin{psmallmatrix} 0&0&0&0\\1&1&0&0\\1&0&1&0\\2&1&1&0\end{psmallmatrix}.
\]
Swapping rows and columns $1,2$ of $\Pi_{C,M}$ gives $\Pi_{C',M}$, i.e.\ $\Pi_{C',M} = P^T\Pi_{C,M}P$: the same permutation $\sigma$ that maps $C$ to $C'$ is an isomorphism of the two weighted graphs. Conversely, any isomorphism between the graphs is a permutation mapping $C$ to~$C'$, since the row space of $\Pi_{C,M}$ is~$C$ (Theorem~\ref{thm:generalized_massey}).
\end{example}

Thus the orthogonal projector method, with any nondegenerate symmetric form that yields a valid reduction, reaches exactly the codes of Theorem~\ref{thm:gi_reducible_count}: one extra dimension of hull, under a condition on the hull vector, and nothing new in characteristic~$2$. Codes with $\dim\Hull(C) \geq 2$, including all self-orthogonal and self-dual codes, are provably out of reach of every symmetric bilinear form.

\section{Reduction Algorithm}
\label{sec:algorithm}

Having characterized GI-reducible codes, we now present the explicit reduction algorithm. Its correctness rests on Theorems~\ref{thm:main_characterization} and~\ref{thm:gi_reducible_count} together with the following two permutation invariants.

\begin{lemma}
\label{lem:invariants}
Let $C$ be an $[n,k]_q$ code and $P$ a permutation matrix. Then $\Hull(CP) = \Hull(C)P$; in particular $\dim\Hull(CP) = \dim\Hull(C)$, and if $\Hull(C) = \langle x\rangle$ then $\Hull(CP) = \langle xP\rangle$ with $(xP, \1) = (x, \1)$.
\end{lemma}
\begin{proof}
$P$ is an isometry of the standard inner product, so $(CP)^\perp = C^\perp P$ and $\Hull(CP) = CP \cap C^\perp P = \Hull(C)P$. Finally $(xP,\1) = xP\1^T = x\1^T = (x,\1)$ since $P\1^T = \1^T$.
\end{proof}

\begin{lemma}
\label{lem:any_b}
Let $q$ be odd and let $C$ be an $[n,k]_q$ code with $\Hull(C) = \langle x\rangle$ and $(x,\1) \neq 0$. Then $C$ is $M$-LCD for \emph{every} nondegenerate $M = aI + bJ$ with $b \neq 0$.
\end{lemma}
\begin{proof}
This is the computation in the proof of Theorem~\ref{thm:gi_reducible_count}: $\det(GMG^T) = b(x,\1)^2 a^{k-1}\det(G'G'^T) \neq 0$ for $a \neq 0$, $b \neq 0$.
\end{proof}

Given two $[n,k]_q$ codes $C_1, C_2$ with generator matrices $G_1, G_2$, the algorithm proceeds as follows; the procedure is summarized as Algorithm~\ref{alg:gi_reduction}.
\begin{enumerate}
\item \emph{Hull dimensions.} Compute $h_i = \dim\Hull(C_i) = k - \rank(G_iG_i^T)$. If $h_1 \neq h_2$, return ``$C_1 \not\cong C_2$'' (Lemma~\ref{lem:invariants}). If $h_1 = h_2 \geq 2$, return ``not GI-reducible'' (Theorem~\ref{thm:gi_reducible_count}).
\item \emph{LCD case.} If $h_1 = h_2 = 0$, set $M = I$ (the reduction of Bardet et al.).
\item \emph{Hull dimension one.} If $h_1 = h_2 = 1$, compute generators $x_i$ of $\Hull(C_i)$ (a nonzero vector $c_iG_i$ with $c_i$ in the left kernel of $G_iG_i^T$) and the invariants $s_i = [(x_i,\1) \neq 0]$. If $s_1 \neq s_2$, return ``$C_1 \not\cong C_2$'' (Lemma~\ref{lem:invariants}). If $s_1 = s_2 = 0$, neither code is GI-reducible (Theorem~\ref{thm:gi_reducible_count}); return ``not GI-reducible''. For even $q$ this branch is always taken (Corollary~\ref{cor:char2}). Otherwise $q$ is odd; choose any $b \in \FF_q^*$ with $1 + nb \neq 0$ (such $b$ exists since $q \geq 3$) and set $M = I + bJ$. By Lemma~\ref{lem:any_b} both codes are $M$-LCD.
\item \emph{Projectors.} Compute $\Pi_i = MG_i^T(G_iMG_i^T)^{-1}G_i$ for $i = 1, 2$. In general $\Pi_i$ is not symmetric (unless $M = I$), so it is read as the weighted adjacency matrix of a directed graph with edge weights in $\FF_q$; isomorphism of such graphs reduces to GI in polynomial time (Section~\ref{sec:preliminaries}).
\item \emph{GI solver.} Decide whether $\Pi_2 = P^T\Pi_1P$ for some permutation matrix~$P$. Return ``$C_1 \cong C_2$'' if so and ``$C_1 \not\cong C_2$'' otherwise; both answers are correct by~(V2) of Theorem~\ref{thm:main_characterization}, since $C_1$ and $C_2$ are $M$-LCD for the same~$M$.
\end{enumerate}

\begin{algorithm}[!htb]
\caption{Reduction from PEP to GI via permutation-invariant bilinear forms}
\label{alg:gi_reduction}
\begin{algorithmic}[1]
\REQUIRE Generator matrices $G_1, G_2 \in \FF_q^{k \times n}$ of $[n,k]_q$ codes $C_1, C_2$.
\ENSURE ``$C_1 \cong C_2$'', ``$C_1 \not\cong C_2$'', or ``not GI-reducible''.
\STATE $h_i \gets k - \rank(G_i G_i^T)$ for $i = 1, 2$ \hfill\COMMENT{hull dimensions}
\IF{$h_1 \neq h_2$}
  \STATE \textbf{return} ``$C_1 \not\cong C_2$''
\ENDIF
\IF{$h_1 \geq 2$}
  \STATE \textbf{return} ``not GI-reducible''
\ENDIF
\IF{$h_1 = 0$}
  \STATE $M \gets I$ \hfill\COMMENT{LCD case: Bardet et al.}
\ELSE
  \STATE $x_i \gets c_i G_i$ with $0 \neq c_i \in \ker_{\mathrm{left}}(G_iG_i^T)$;\quad $s_i \gets [\,x_i\1^T \neq 0\,]$ for $i = 1,2$
  \IF{$s_1 \neq s_2$}
    \STATE \textbf{return} ``$C_1 \not\cong C_2$'' \hfill\COMMENT{Lemma~\ref{lem:invariants}}
  \ENDIF
  \IF{$s_1 = s_2 = 0$}
    \STATE \textbf{return} ``not GI-reducible'' \hfill\COMMENT{Theorem~\ref{thm:gi_reducible_count}}
  \ENDIF
  \STATE choose $b \in \FF_q^*$ with $1 + nb \neq 0$;\quad $M \gets I + bJ$ \hfill\COMMENT{Lemma~\ref{lem:any_b}}
\ENDIF
\STATE $\Pi_i \gets M G_i^T (G_i M G_i^T)^{-1} G_i$ for $i = 1, 2$
\STATE decide $\Pi_1 \cong \Pi_2$ as edge-weighted directed graphs via a GI solver
\IF{$\Pi_1 \cong \Pi_2$}
  \STATE \textbf{return} ``$C_1 \cong C_2$''
\ELSE
  \STATE \textbf{return} ``$C_1 \not\cong C_2$''
\ENDIF
\end{algorithmic}
\end{algorithm}

\begin{theorem}[Correctness]
\label{thm:alg_correct}
Algorithm~\ref{alg:gi_reduction} returns ``not GI-reducible'' only if neither $C_1$ nor $C_2$ is GI-reducible, and whenever it returns a verdict ``$C_1 \cong C_2$'' or ``$C_1 \not\cong C_2$'' that verdict is correct.
\end{theorem}
\begin{proof}
If $h_1 \neq h_2$, or $h_1 = h_2 = 1$ with $s_1 \neq s_2$, the codes are inequivalent by Lemma~\ref{lem:invariants}, and the algorithm says so (this verdict is correct whether or not the codes are GI-reducible). If $h_1 = h_2 \geq 2$, or $h_1 = h_2 = 1$ with $s_1 = s_2 = 0$, neither code is GI-reducible by Theorem~\ref{thm:gi_reducible_count}, and the algorithm returns ``not GI-reducible''. In the remaining cases both codes are $M$-LCD for the chosen $M = aI + bJ$ (trivially for $M = I$; by Lemma~\ref{lem:any_b} otherwise), and condition~(V2), which holds by Theorem~\ref{thm:main_characterization}, shows that the GI query decides $C_1 \cong C_2$.
\end{proof}

\begin{proposition}[Complexity]
\label{prop:complexity}
Algorithm~\ref{alg:gi_reduction} performs $O(n^3)$ field operations and one GI query on a graph with $N = O(n^2\log q)$ vertices; with the algorithm of Babai~\cite{babai2016} it therefore runs in time $2^{O(\log^cN)}$ for a constant~$c$, which is quasi-polynomial in $n\log q$.
\end{proposition}
\begin{proof}
Lines~1 and~11 compute $G_iG_i^T$, its rank and a vector of its left kernel in $O(k^2n + k^3)$ operations, and line~20 computes $G_iMG_i^T$, its inverse and $\Pi_i$ in $O(k^2n + k^3 + kn^2)$ operations; with $k \leq n$ this is $O(n^3)$. A single structure matrix is used, so no search over the parameter~$b$ is needed. The edge-weighted directed graphs $\Pi_1, \Pi_2$ on $n$ vertices are converted by the standard gadget into unweighted undirected graphs on $O(n^2\log q)$ vertices~\cite{bardet2019permutation}, to which the GI solver is applied once.
\end{proof}

A SageMath implementation of Algorithm~\ref{alg:gi_reduction} (Appendix~\ref{app:script}) was compared with a brute-force permutation-equivalence test on \emph{all} pairs of $[n,k]_q$ codes of the same dimension $1 \leq k \leq n-1$, for
\[
(q,n) \in \{(2,3),(3,3),(5,3),(2,4),(3,4),(5,4),(2,5),(3,5)\},
\]
$1{,}867{,}212$ pairs in total. Every verdict ``$C_1 \cong C_2$'' or ``$C_1 \not\cong C_2$'' agrees with the brute-force test, and ``not GI-reducible'' is returned only on pairs in which neither code is GI-reducible.

\section{Counting GI-Reducible Codes}
\label{sec:enumeration}

Having characterized which codes are GI-reducible (Theorem~\ref{thm:gi_reducible_count}), we now enumerate them, with two aims: an exact formula, and a quantitative comparison with the LCD codes reached by Bardet et al. By Corollary~\ref{cor:char2}, the GI-reducible codes over $\FF_{2^m}$ are precisely the LCD codes, whose enumeration is given in~\cite{carlet2019new,lishiling2025} (see Appendix~\ref{app:lcd_enumeration}). This section therefore assumes $q$ odd; the algorithm of Section~\ref{sec:algorithm} is stated for general $q$.

By the orthogonal decomposition $C = \Hull(C) \perp U$ (Section~\ref{sec:characterization}), the count splits into LCD codes and hull dimension 1 codes with $(x,\1) \neq 0$. Formulas for the number $L(n,k,q)$ of LCD codes are known (Appendix~\ref{app:lcd_enumeration}); the key is computing $K(n,q) = |\{x \in \FF_q^n : (x,x) = 0, (x,\1) \neq 0\}|$. We compute this using Weil's theorem on quadratic forms (Theorem~\ref{thm:weil_lidl_niederreiter}).

\begin{lemma}[Isotropic Vectors with Nonzero Coordinate Sum]
\label{lem:quadratic_count}
Let $q$ be odd with characteristic $p$, and let $n \geq 2$ be a positive integer.
The number $K(n,q)$ of vectors $x \in \mathbb{F}_q^n$ with $(x,x) = 0$ and $(x, \1) \neq 0$ is given by
$$
K(n,q) = \begin{cases}
q^{n-2}(q-1), & p \mid n, \\
q^{n-2}(q-1) - \chi((-1)^{(n-1)/2} n)(q-1)q^{(n-3)/2}, & p \nmid n,\ n \text{ odd,} \\
q^{n-2}(q-1) + \chi((-1)^{n/2})(q-1)q^{(n-2)/2}, & p \nmid n,\ n \text{ even,}
\end{cases}
$$
where $\chi$ is the quadratic character of $\mathbb{F}_q$.
\end{lemma}

\begin{proof}
We compute $K(n,q) = N_0(n,q) - N_1(n,q)$, where $N_0(n,q) = |\{x \in \FF_q^n : (x,x) = 0\}|$ and $N_1(n,q) = |\{x \in \FF_q^n : (x,x) = 0, (x,\1) = 0\}|$.

Applying Theorem~\ref{thm:weil_lidl_niederreiter} to $Q(x) = \sum_{i=1}^n x_i^2$ (which is nondegenerate with $\det(Q) = 1$) gives
$$
N_0(n,q) = \begin{cases}
q^{n-1} & \text{if } n \text{ is odd,} \\
q^{n-1} + \chi((-1)^{n/2})(q-1)q^{(n-2)/2} & \text{if } n \text{ is even.}
\end{cases}
$$

For $N_1(n,q)$, we count solutions on the hyperplane $H = \{x : (x,\1) = 0\}$ of dimension $n-1$. Since $H^\perp = \langle \1 \rangle$ and $(\1,\1) = n$, the radical of $Q|_H$ is $H \cap H^\perp = H \cap \langle \1 \rangle$, which is $\langle \1 \rangle$ if $p \mid n$ and $\{0\}$ if $p \nmid n$.

\emph{Case $p \nmid n$}: $Q|_H$ is nondegenerate of dimension $n-1$. To apply Weil's theorem, we need $\det(Q|_H)$. Since $\FF_q^n = H \perp \langle \1 \rangle$ and $Q(\1) = n$, we have $\det(Q) = \det(Q|_H) \cdot n \cdot d^2$ for some nonzero $d$, so $\chi(\det(Q|_H)) = \chi(\det(Q)) \cdot \chi(n)^{-1} = \chi(n)$. Thus
$$
N_1(n,q) = \begin{cases}
q^{n-2} & \text{if } n \text{ is even,} \\
q^{n-2} + \chi((-1)^{(n-1)/2} n)(q-1)q^{(n-3)/2} & \text{if } n \text{ is odd.}
\end{cases}
$$

\emph{Case $p \mid n$}: $\1 \in H$ and $Q|_H$ has radical $\langle \1 \rangle$. The quotient form $\bar{Q}$ on $\bar{H} = H/\langle \1 \rangle$ is nondegenerate of dimension $n-2$. Since $Q(x + t\1) = Q(x)$ for all $x \in H$ and $t \in \FF_q$, each zero of $\bar{Q}$ lifts to $q$ zeros in $H$. In the basis $\{\bar{g}_i = \overline{e_i - e_n} : i = 1, \ldots, n-2\}$ of $\bar{H}$, the matrix of $\bar{Q}$ is $I_{n-2} + J_{n-2}$ (diagonal entries~$2$, off-diagonal entries~$1$). By the determinant formula in Remark~\ref{rem:nondegeneracy} applied with $a = b = 1$ and dimension $n-2$, $\det(I_{n-2} + J_{n-2}) = 1^{n-3}(1 + (n-2)) = n - 1 = -1$ in $\FF_q$ (using $p \mid n$). Applying Weil's theorem and multiplying by $q$:
$$
N_1(n,q) = \begin{cases}
q^{n-2} & \text{if } n \text{ is odd,} \\
q^{n-2} + \chi((-1)^{n/2})(q-1)q^{(n-2)/2} & \text{if } n \text{ is even,}
\end{cases}
$$
where we used $\chi((-1)^{(n-2)/2} \cdot (-1)) = \chi((-1)^{n/2})$.

Subtracting $N_1(n,q)$ from $N_0(n,q)$ yields the stated formula.
\end{proof}

To count hull-dimension-1 codes, we fix a hull line $\langle x \rangle$ and enumerate codes with that hull. The structure of the residual bilinear form on $V = x^\perp / \langle x \rangle$ governs this count.

\begin{lemma}[Hull-Line Counting]\label{lem:hull_line_count}
Let $x \in \FF_q^n$ satisfy $(x,x) = 0$ and $(x,\1) \neq 0$. The number of $[n,k]_q$ codes $C$ with $\Hull(C) = \langle x \rangle$ equals $L^{\chi(-1)}(n-2,\, k-1,\, q)$, where $L^\epsilon(n,k,q)$ for $\epsilon \in \{+1, -1\}$ denotes the number of $k$-dimensional subspaces of $\FF_q^n$ on which a nondegenerate symmetric bilinear form of type~$\epsilon$ restricts to a nondegenerate form (Appendix~\ref{app:type_epsilon}).
\end{lemma}
\begin{proof}
Since $(x,x) = 0$, the standard inner product restricted to $x^\perp$ has radical $x^\perp \cap \langle x\rangle = \langle x \rangle$, and induces a nondegenerate form on the quotient $V := x^\perp / \langle x \rangle$ of dimension $n - 2$. By Lemma~\ref{lem:hyperbolic_completion}, we may choose $v$ such that $\{x, v\}$ forms a hyperbolic plane: $(x, x) = (v, v) = 0$ and $(x, v) = 1$. Put $W := \langle x, v\rangle^\perp$. Then $W \subseteq x^\perp$, $W \cap \langle x\rangle = \{0\}$ since $(x,v) = 1$, and $\dim W = n - 2$. The projection $x^\perp \to V$ preserves the form, is injective on~$W$ and hence, by dimension, restricts to an isometry $W \to V$. The orthogonal decomposition $\FF_q^n = \langle x, v\rangle \perp \langle x, v\rangle^\perp$ together with $\disc(\langle x, v\rangle) = -1$ then yields, by multiplicativity of discriminants,
\[
\disc(V) = \disc(\FF_q^n)/\disc(\langle x, v \rangle) = -1 \quad \text{in } \FF_q^*/(\FF_q^*)^2.
\]

A code $C$ with $\Hull(C) = \langle x \rangle$ satisfies $\langle x \rangle \subseteq C \subseteq x^\perp$ (the first inclusion since $x \in C$; the second since $x \in \Hull(C) \subseteq C^\perp$ implies $C \subseteq x^\perp$). Mapping $C \mapsto \bar C := C/\langle x \rangle \subseteq V$ gives a bijection between codes $C$ with $\langle x \rangle \subseteq C \subseteq x^\perp$ and subspaces $\bar C \subseteq V$, with $\dim \bar C = k - 1$.

Under this bijection, the radical of the form of $V$ restricted to $\bar C$ equals $\Hull(C)/\langle x \rangle$: for $u \in C$, the class $\bar u \in \bar C$ lies in this radical if and only if $(\bar u, \bar c) = (u, c) = 0$ for all $c \in C$ (the induced form is well-defined modulo the radical $\langle x \rangle$ of $x^\perp$), which is equivalent to $u \in C^\perp$, and hence to $u \in \Hull(C)$. Therefore, $\Hull(C) = \langle x \rangle$ if and only if the form of $V$ restricts to a nondegenerate form on $\bar C$. Since $V$ has discriminant $-1$ in $\FF_q^*$, its isometry type is $\epsilon = \chi(-1)$, and the count of such $\bar C$ is $L^{\chi(-1)}(n-2, k-1, q)$, given in Appendix~\ref{app:type_epsilon}.
\end{proof}

\begin{theorem}[Complete Enumeration of GI-Reducible Codes]\label{thm:enumeration_main}
Let $q$ be odd, and let $n \geq 3$ and $1 \leq k < n$. The number of GI-reducible $[n,k]_q$ codes is
\[
L(n,k,q) + \frac{K(n,q)}{q-1} \cdot L^{\chi(-1)}(n-2,\, k-1,\, q),
\]
where $L(n,k,q)$ counts standard LCD codes (Appendix~\ref{app:lcd_enumeration}), $K(n,q)$ counts isotropic vectors with nonzero coordinate sum (Lemma~\ref{lem:quadratic_count}), and $L^\epsilon(n,k,q)$ counts the $k$-dimensional subspaces on which a type-$\epsilon$ form is nondegenerate (Appendix~\ref{app:type_epsilon}). For $k \in \{1, n-1\}$ the conventions $L^\epsilon(m,0,q) = L^\epsilon(m,m,q) = 1$ ($m \geq 1$) apply.
\end{theorem}

\begin{proof}
By Theorem~\ref{thm:gi_reducible_count}, GI-reducible codes split into LCD codes (count $L(n,k,q)$) and codes with $\dim(\Hull(C)) = 1$ whose hull vector $x$ satisfies $(x,\1) \neq 0$. The valid hull lines $\langle x \rangle$ correspond to projective points: the set $\{x \in \FF_q^n : (x,x)=0,\, (x,\1) \neq 0\}$ has $K(n,q)$ elements, partitioned into $K(n,q)/(q-1)$ lines (each line contains $q-1$ generators in this set, since the conditions are scalar-invariant). For each line, Lemma~\ref{lem:hull_line_count} gives $L^{\chi(-1)}(n-2, k-1, q)$ codes with that hull.
\end{proof}

\section{Coverage of the Orthogonal Projector Method}
\label{sec:coverage}

\begin{table}[t]
\centering
\caption{Fractions $p_{\mathrm{LCD}}$ (codes reached by Bardet et al.) and $p_{\mathrm{GI}}$ (codes reached by the orthogonal projector method, Theorem~\ref{thm:enumeration_main}) of all $[n,k]_q$ codes, and their ratio, for small parameters; the last row is the limit of Proposition~\ref{prop:coverage}.}
\label{tab:coverage}
\footnotesize\setlength{\tabcolsep}{4pt}
\begin{tabular}{cc|ccc|ccc|ccc}
\hline
 & & \multicolumn{3}{c|}{$q=3$} & \multicolumn{3}{c|}{$q=5$} & \multicolumn{3}{c}{$q=7$} \\
$n$ & $k$ & $p_{\mathrm{LCD}}$ & $p_{\mathrm{GI}}$ & ratio & $p_{\mathrm{LCD}}$ & $p_{\mathrm{GI}}$ & ratio & $p_{\mathrm{LCD}}$ & $p_{\mathrm{GI}}$ & ratio \\
\hline
4 & 1 & 0.6000 & 0.9000 & 1.5000 & 0.7692 & 0.9615 & 1.2500 & 0.8400 & 0.9800 & 1.1667 \\
4 & 2 & 0.6923 & 0.8769 & 1.2667 & 0.8065 & 0.9553 & 1.1846 & 0.8596 & 0.9775 & 1.1371 \\
6 & 1 & 0.6923 & 0.9148 & 1.3214 & 0.7937 & 0.9601 & 1.2097 & 0.8596 & 0.9796 & 1.1395 \\
6 & 2 & 0.6694 & 0.8901 & 1.3297 & 0.8003 & 0.9537 & 1.1917 & 0.8572 & 0.9771 & 1.1399 \\
6 & 3 & 0.6694 & 0.8846 & 1.3214 & 0.7876 & 0.9527 & 1.2097 & 0.8572 & 0.9768 & 1.1395 \\
8 & 1 & 0.6585 & 0.8890 & 1.3500 & 0.7987 & 0.9600 & 1.2019 & 0.8568 & 0.9796 & 1.1433 \\
8 & 2 & 0.6670 & 0.8644 & 1.2959 & 0.8000 & 0.9536 & 1.1920 & 0.8571 & 0.9771 & 1.1399 \\
8 & 4 & 0.6449 & 0.8547 & 1.3252 & 0.7939 & 0.9521 & 1.1994 & 0.8547 & 0.9767 & 1.1427 \\
12 & 1 & 0.6658 & 0.8880 & 1.3338 & 0.7999 & 0.9600 & 1.2001 & 0.8571 & 0.9796 & 1.1429 \\
12 & 2 & 0.6667 & 0.8634 & 1.2951 & 0.8000 & 0.9536 & 1.1920 & 0.8571 & 0.9771 & 1.1399 \\
12 & 6 & 0.6397 & 0.8514 & 1.3311 & 0.7934 & 0.9520 & 1.2000 & 0.8546 & 0.9767 & 1.1429 \\
16 & 1 & 0.6666 & 0.8889 & 1.3335 & 0.8000 & 0.9600 & 1.2000 & 0.8571 & 0.9796 & 1.1429 \\
16 & 2 & 0.6667 & 0.8642 & 1.2963 & 0.8000 & 0.9536 & 1.1920 & 0.8571 & 0.9771 & 1.1399 \\
16 & 8 & 0.6391 & 0.8520 & 1.3332 & 0.7933 & 0.9520 & 1.2000 & 0.8546 & 0.9767 & 1.1429 \\
$\infty$ & $\sim n/2$ & 0.6390 & 0.8520 & 1.3333 & 0.7933 & 0.9520 & 1.2000 & 0.8546 & 0.9767 & 1.1429 \\
\hline
\end{tabular}
\end{table}

Let $\mathcal{G}(n,k,q) = \left[\begin{smallmatrix} n \\ k\end{smallmatrix}\right]_q$ denote the number of $[n,k]_q$ codes, and write
\begin{align*}
p_{\mathrm{LCD}}(n,k,q) &:= \frac{L(n,k,q)}{\mathcal{G}(n,k,q)}, \\
p_{\mathrm{GI}}(n,k,q) &:= \frac{L(n,k,q) + \frac{K(n,q)}{q-1} L^{\chi(-1)}(n-2,k-1,q)}{\mathcal{G}(n,k,q)}
\end{align*}
for the fractions of $[n,k]_q$ codes reached by the reduction of Bardet et al.\ and by the orthogonal projector method, respectively.

Sendrier~\cite{sendrier1997hull} showed that the fraction of codes with hull dimension~$h$ converges, as $k$ and $n-k$ tend to infinity, to a constant depending only on $q$ and~$h$; for $h = 0$ the limit is $\prod_{i \geq 1}(1 + q^{-i})^{-1}$ (see also~\cite{carlet2019new}). The following proposition quantifies the gain of GI-reducibility over LCD-ness in the same regime.

\begin{proposition}[Asymptotic coverage]
\label{prop:coverage}
Let $q$ be odd. As $k \to \infty$ and $n - k \to \infty$,
\[
p_{\mathrm{GI}}(n,k,q) \;\longrightarrow\; \Bigl(1 + \frac{1}{q}\Bigr) \prod_{i \geq 1}\bigl(1 + q^{-i}\bigr)^{-1} \;=\; \prod_{i \geq 2}\bigl(1 + q^{-i}\bigr)^{-1}.
\]
\end{proposition}

\begin{proof}
Write $P_\infty(q) := \prod_{i \geq 1}(1+q^{-i})^{-1}$. By Lemma~\ref{lem:quadratic_count}, $K(n,q)/(q-1) = q^{n-2}(1 + O(q^{-(n-3)/2}))$. By Theorem~\ref{thm:Lepsilon}, $L^{\epsilon}(n-2,k-1,q)/L(n-2,k-1,q)$ equals~$1$ or $(q^{(n-2)/2} + \epsilon')/(q^{(n-2)/2} - \epsilon')$ with $\epsilon' = \pm 1$, hence tends to~$1$; together with Sendrier's limit, $L^{\chi(-1)}(n-2,k-1,q)/\mathcal{G}(n-2,k-1,q) \to P_\infty(q)$. Finally,
\[
\frac{\mathcal{G}(n-2,k-1,q)}{\mathcal{G}(n,k,q)} = \frac{(q^k - 1)(q^{n-k} - 1)}{(q^n - 1)(q^{n-1} - 1)} = q^{-(n-1)}\bigl(1 + O(q^{-k}) + O(q^{-(n-k)})\bigr).
\]
Multiplying the three factors gives $\frac{K(n,q)}{q-1}\, L^{\chi(-1)}(n-2,k-1,q)/\mathcal{G}(n,k,q) \to q^{-1}P_\infty(q)$, and adding $p_{\mathrm{LCD}} \to P_\infty(q)$ yields the claim. The product identity follows from $(1 + q^{-1})^{-1}\prod_{i\geq 1}(1+q^{-i}) = \prod_{i \geq 2}(1 + q^{-i})$.
\end{proof}

The convergence in Proposition~\ref{prop:coverage} is fast: for $k = \lfloor n/2 \rfloor$ and $q \in \{3,5,7\}$ the exact values of $p_{\mathrm{GI}}$ computed from Theorem~\ref{thm:enumeration_main} are within $3 \cdot 10^{-3}$ of the limit for $n = 8$, within $10^{-3}$ for $n = 12$, and agree with it to four decimal places for $n = 16$ (Table~\ref{tab:coverage}). The values in Table~\ref{tab:coverage} were computed from Theorem~\ref{thm:enumeration_main} and cross-checked against exhaustive enumeration for $n \leq 6$.

For fixed dimension $k$ and $n \to \infty$ the limits are different, and can be read off the exact formulas.

\begin{proposition}[Fixed dimension]
\label{prop:fixedk}
Let $q$ be odd and $k \geq 1$ be fixed, and put
\[
P_k(q) := \prod_{i=1}^{\lceil k/2 \rceil}\bigl(1 - q^{1-2i}\bigr), \qquad P_0(q) := 1 .
\]
Then, as $n \to \infty$,
\[
p_{\mathrm{LCD}}(n,k,q) \to P_k(q), \qquad
p_{\mathrm{GI}}(n,k,q) \to P_k(q) + \frac{1 - q^{-k}}{q}\,P_{k-1}(q).
\]
Consequently $p_{\mathrm{GI}}/p_{\mathrm{LCD}} \to 1 + q^{-1}$ for odd~$k$ and $\to 1 + (1 - q^{-k})q^{-1}$ for even~$k$.
\end{proposition}

\begin{proof}
Write $\mathcal{G}(N,K,Q) = \left[\begin{smallmatrix} N \\ K\end{smallmatrix}\right]_Q = Q^{K(N-K)}\prod_{i=1}^{K}\frac{1 - Q^{-(N-K+i)}}{1 - Q^{-i}}$, so that for fixed~$K$, $\mathcal{G}(N,K,Q) = Q^{K(N-K)}\prod_{i=1}^K(1-Q^{-i})^{-1}(1 + o(1))$ as $N \to \infty$. Insert this into the four cases of Carlet's formula for $L(n,k,q)$ (Appendix~\ref{app:mass_odd}). In each case the power of~$q$ in $L(n,k,q)$ combines to $q^{k(n-k)}$: for instance, for $k$ odd and $n$ even the exponent is $\frac{k(n-k)-1}{2} + \frac{n}{2} + 2\cdot\frac{k-1}{2}\bigl(\frac{n}{2} - 1 - \frac{k-1}{2}\bigr) = k(n-k)$, and the other three cases are similar. The remaining factor is $\prod_{i=1}^{k/2}(1 - q^{-2i})^{-1}$ when $k$ is even and $\prod_{i=1}^{(k-1)/2}(1-q^{-2i})^{-1}$ when $k$ is odd, independently of the parity of~$n$ (for $k$ odd, $n$ even the factor $q^{n/2} - \chi((-1)^{n/2})$ contributes $q^{n/2}(1 + o(1))$). Dividing by $\mathcal{G}(n,k,q) = q^{k(n-k)}\prod_{i=1}^k(1-q^{-i})^{-1}(1+o(1))$ cancels the even-indexed factors and leaves $\prod_{i \text{ odd}, i \leq k}(1 - q^{-i}) = P_k(q)$.

For the second term of $p_{\mathrm{GI}}$, Lemma~\ref{lem:quadratic_count} gives $K(n,q)/(q-1) = q^{n-2}(1+o(1))$; Theorem~\ref{thm:Lepsilon} gives $L^{\chi(-1)}(n-2,k-1,q)/L(n-2,k-1,q) \to 1$ (the ratio is $1$ or $(q^{(n-2)/2} \pm 1)/(q^{(n-2)/2} \mp 1)$); the first part applied to $k-1$ gives $L(n-2,k-1,q)/\mathcal{G}(n-2,k-1,q) \to P_{k-1}(q)$; and
\[
\frac{\mathcal{G}(n-2,k-1,q)}{\mathcal{G}(n,k,q)} = \frac{(q^k-1)(q^{n-k}-1)}{(q^n-1)(q^{n-1}-1)} = \frac{q^k - 1}{q^{\,n+k-1}}\,(1 + o(1)).
\]
The product of the four factors tends to $q^{n-2} \cdot P_{k-1}(q) \cdot (q^k-1)q^{-(n+k-1)} = (1-q^{-k})P_{k-1}(q)/q$. The ratio statement follows from $P_k = P_{k-1}$ for even~$k$ and $P_k = (1 - q^{-k})P_{k-1}$ for odd~$k$.
\end{proof}

\section{Conclusion}
\label{sec:conclusion}

In this paper we determined exactly how far the LCD-to-GI reduction of Bardet et al.\ can be pushed by replacing the standard inner product with a nondegenerate symmetric bilinear form~$M$. For $n \geq 3$ and every~$q$, the $M$-orthogonal projector is a reduction of PEP to GI if and only if $M$ is permutation-invariant, i.e.\ $M = aI + bJ$ (Theorem~\ref{thm:main_characterization}). The codes reached in this way are exactly those whose hull has dimension at most one and, for a hull line $\langle x\rangle$, $(x,\1) \neq 0$ (Theorem~\ref{thm:gi_reducible_count}); in characteristic~$2$ only LCD codes are reached. A polynomial-time algorithm implements the reduction (Algorithm~\ref{alg:gi_reduction}). The reducible codes are counted exactly (Theorem~\ref{thm:enumeration_main}): for odd~$q$ the method covers asymptotically a fraction $\prod_{i \geq 2}(1 + q^{-i})^{-1}$ of all codes, a factor $1 + q^{-1}$ over the original reduction (Proposition~\ref{prop:coverage}), approached already at small lengths (Table~\ref{tab:coverage}).

Regarding the classification of linear codes, Theorem~\ref{thm:gi_reducible_count} delimits the codes for which GI solvers can be used directly through a complete invariant: all LCD codes and the codes with a one-dimensional hull whose generator has nonzero coordinate sum, which is the vast majority of codes once $q \geq 5$. On the complexity side, PEP restricted to these codes is polynomial-time reducible to GI and hence solvable in quasi-polynomial time, so any hardness of PEP beyond that of GI must reside in the codes with hull dimension at least two, among them all self-orthogonal and self-dual codes, and in the codes with a hull line of zero coordinate sum; the latter set is nonempty for every~$q$, and in characteristic~$2$ it contains every code with a one-dimensional hull. For code-based cryptography this identifies the easy instances of PEP: schemes whose security rests on PEP should draw their instances from codes with hull dimension at least two, as SPECK~\cite{SPECK} does with self-dual codes.

The central open direction is to go beyond bilinear forms. The projector $\Pi_{C,M}$ is a degree-two polynomial invariant of the generator matrix; whether an equivariant invariant of higher degree, or one built from a basis of the hull, reduces PEP to GI on codes with $\dim\Hull(C) \geq 2$, or whether such reductions are impossible in general, is the central open question.

\section*{Acknowledgements}
K.~Ishizuka was supported by JSPS KAKENHI Grant Number JP25K17290.

\appendix

\section{LCD Code Enumeration}
\label{app:lcd_enumeration}

This appendix presents the mass formulas for LCD codes, which are central to counting GI-reducible codes. We first review the historical development of LCD code enumeration, then provide explicit formulas separated by field characteristic.

\subsection{Historical Development}

The study of LCD codes was initiated by Massey~\cite{massey1992linear}, who introduced the concept and proved that an $[n,k]$ code $C$ is LCD if and only if its generator matrix $G$ satisfies $\det(GG^T) \neq 0$. Massey's work established the foundational characterization but did not address enumeration.

Closed-form counting formulas for LCD codes were derived by Carlet et al.~\cite{carlet2019new}, who used the action of the orthogonal group to derive closed-form counting formulas for LCD codes over $\mathbb{F}_2$ and over finite fields of odd characteristic. Their approach exploited the structure of the orthogonal group and its orbits on the Grassmannian of linear codes. As explained in Appendix~\ref{app:type_epsilon}, these formulas are special cases of the classical Anzahl theorems~\cite{daifeng1965,fengdai1965,wan1993,wan2002}.

For fields of even characteristic beyond $\mathbb{F}_2$, Li et al.~\cite{lishiling2025} derived mass formulas for linear codes with arbitrary prescribed hull dimension $l$ over finite fields of characteristic 2. In particular, for LCD codes (the case $l=0$), their result extends Carlet's formula from $q=2$ to arbitrary $q=2^m$ by observing that the key structural properties of Carlet's construction generalize naturally to all fields of characteristic 2.

\subsection{Mass Formulas for Odd Characteristic}
\label{app:mass_odd}

Let $q = p^m$ where $p$ is an odd prime. Let $\left[\begin{smallmatrix} n \\ k \end{smallmatrix}\right]_q$ $= \prod_{i=1}^{k}\frac{q^{n-k+i}-1}{q^{i}-1}$ denote the Gaussian binomial coefficient, and let $\chi$ denote the quadratic character of $\mathbb{F}_q$. The number $L(n,k,q)$ of LCD codes of type $[n,k]_q$ is given by Carlet et al.~\cite[Corollary~32]{carlet2019new}:

$$
L(n,k,q) = \begin{cases}
q^{\frac{k(n-k)-1}{2}} \left( q^{\frac{n}{2}} - \chi((-1)^{\frac{n}{2}}) \right) \left[\begin{smallmatrix} \frac{n}{2}-1 \\ \frac{k-1}{2} \end{smallmatrix}\right]_{q^2}, & k \text{ odd, } n \text{ even,} \\[6pt]
q^{\frac{(k+1)(n-k)}{2}} \left[\begin{smallmatrix} \frac{n-1}{2} \\ \frac{k-1}{2} \end{smallmatrix}\right]_{q^2}, & k \text{ odd, } n \text{ odd,} \\[6pt]
q^{\frac{k(n-k+1)}{2}} \left[\begin{smallmatrix} \frac{n-1}{2} \\ \frac{k}{2} \end{smallmatrix}\right]_{q^2}, & k \text{ even, } n \text{ odd,} \\[6pt]
q^{\frac{k(n-k)}{2}} \left[\begin{smallmatrix} \frac{n}{2} \\ \frac{k}{2} \end{smallmatrix}\right]_{q^2}, & k \text{ even, } n \text{ even.}
\end{cases}
$$

\subsection{Mass Formulas for Even Characteristic}

Let $q = 2^m$. Carlet et al.~\cite[Corollary~17]{carlet2019new} derived mass formulas for LCD codes over $\mathbb{F}_2$. Li et al.~\cite[Theorem~2 and subsequent remark]{lishiling2025} showed that Carlet's result for LCD codes over $\mathbb{F}_2$ can be naturally extended to LCD codes over $\mathbb{F}_{2^m}$. The number $L(n,k,q)$ of LCD codes of type $[n,k]_q$ depends on the parity of $n$ and $k$:

$$
L(n,k,q) = \begin{cases}
q^{\frac{nk-k^2+n-1}{2}} \left[\begin{smallmatrix} \frac{n}{2}-1 \\ \frac{k-1}{2} \end{smallmatrix}\right]_{q^2}, & k \text{ odd, } n \text{ even,} \\[6pt]
q^{\frac{(n-k)(k+1)}{2}} \left[\begin{smallmatrix} \frac{n-1}{2} \\ \frac{k-1}{2} \end{smallmatrix}\right]_{q^2}, & k \text{ odd, } n \text{ odd,} \\[6pt]
q^{\frac{k(n-k+1)}{2}} \left[\begin{smallmatrix} \frac{n-1}{2} \\ \frac{k}{2} \end{smallmatrix}\right]_{q^2}, & k \text{ even, } n \text{ odd,} \\[6pt]
q^{\frac{k(n-k)}{2}} \left( q^{n-k} \left[\begin{smallmatrix} \frac{n}{2}-1 \\ \frac{k}{2}-1 \end{smallmatrix}\right]_{q^2} + \left[\begin{smallmatrix} \frac{n}{2}-1 \\ \frac{k}{2} \end{smallmatrix}\right]_{q^2} \right), & k \text{ even, } n \text{ even.}
\end{cases}
$$
\section{Mass Formulas for Type-$\epsilon$ Forms}
\label{app:type_epsilon}

\subsection{Motivation}

The counting of subspaces of a finite orthogonal space according to the isometry type of the restricted form is classical: these are the \emph{Anzahl theorems} of Dai, Feng and Wan~\cite{daifeng1965,fengdai1965}, collected in Wan's monograph~\cite{wan1993,wan2002} (see also the survey~\cite{wan1997}). The LCD mass formulas of Carlet et al.~\cite{carlet2019new} and Li et al.~\cite{lishiling2025} are rediscoveries of special cases of the Anzahl theorems, derived under the implicit assumption that the ambient form is the standard inner product, i.e.\ of type $\epsilon = +1$. In the Hull-Line Counting Lemma (Lemma~\ref{lem:hull_line_count}), however, the form induced on the quotient $V = x^\perp / \langle x \rangle$ has discriminant $-1$, which corresponds to type $\epsilon = \chi(-1)$: namely $\epsilon = +1$ when $q \equiv 1 \pmod 4$ and $\epsilon = -1$ when $q \equiv 3 \pmod 4$. The case $\epsilon = -1$ does not appear in the coding-theoretic literature; it is needed in Lemma~\ref{lem:hull_line_count}, and hence in Theorem~\ref{thm:enumeration_main}. This appendix therefore restates the count for both types in the notation and context of coding theory and gives a self-contained proof. It is placed as a reference independent of the main text, also in acknowledgement of the Anzahl theorems, of which it is a special case.

\subsection{Setup and Diagonalisation}

By the classification of nondegenerate symmetric bilinear forms over $\FF_q$ ($q$ odd) recalled in Section~\ref{sec:preliminaries}, every such form on $\FF_q^n$ is isometric to exactly one of
\[
M^{(+1)} := \mathrm{diag}(1, 1, \ldots, 1, 1) \quad \text{or} \quad M^{(-1)} := \mathrm{diag}(1, 1, \ldots, 1, \gamma),
\]
where $\gamma \in \FF_q^*$ is a fixed non-square. We refer to $\epsilon \in \{+1, -1\}$ as the \emph{type} of the form, and adopt $M^{(\epsilon)}$ as the structure matrix throughout this appendix.

For $\epsilon \in \{+1, -1\}$, with $M^{(\epsilon)}$ as above, a subspace $W \subseteq \FF_q^n$ is \emph{nondegenerate} if the restriction of $M^{(\epsilon)}$ to $W$ is nondegenerate, i.e.\ $\det(G M^{(\epsilon)} G^T) \neq 0$ for some (any) matrix~$G$ whose rows form a basis of~$W$; as a code, $W$ is then $M^{(\epsilon)}$-LCD (Section~\ref{sec:extension}). Define
\[
L^\epsilon(n, k, q) := \bigl|\{W \subseteq \FF_q^n : \dim W = k,\ W \text{ is nondegenerate}\}\bigr|.
\]
The case $\epsilon = +1$ recovers the standard LCD count $L(n, k, q)$ of Appendix~\ref{app:lcd_enumeration}.

A two-fold classification is at play: while $M^{(\epsilon)}$ itself has type~$\epsilon$, each nondegenerate subspace~$W$ has its own \emph{inner discriminant class} $\delta_W \in \{1, \gamma\}$, defined by the property that some basis matrix~$G$ of~$W$ satisfies $G M^{(\epsilon)} G^T = \mathrm{diag}(1, \ldots, 1, \delta_W)$ (existence by~\cite[Theorem~25]{carlet2019new}). Correspondingly, we partition
\[
L^\epsilon(n, k, q) = |\mathcal{N}^\epsilon_+[n, k]_q| + |\mathcal{N}^\epsilon_-[n, k]_q|,
\]
where the subscript indicates $\delta_W = 1$ or $\delta_W = \gamma$, respectively.

\subsection{Reference Subspaces and Orbit Decomposition}

Let $1 \leq k < n$ and assume $M^{(\epsilon)} = \mathrm{diag}(1, \ldots, 1, \tau_\epsilon)$ with $\tau_{+1} := 1$, $\tau_{-1} := \gamma$. Denote the standard basis of $\FF_q^n$ by $e_1, \ldots, e_n$. Recall (\cite[Lemma~26]{carlet2019new}; also a direct counting argument since the squares form a subgroup of index $2$ in $\FF_q^*$ for $q$ odd) that every element of $\FF_q$ is a sum of two squares; fix $\alpha, \beta \in \FF_q$ with $\alpha^2 + \beta^2 = \gamma$.

\begin{lemma}[Reference subspaces]\label{lem:reference_codes}
Let $1 \leq k < n$ and $\epsilon \in \{+1, -1\}$, with $M^{(\epsilon)}$ as above. Then:
\begin{enumerate}
\item[(i)] The subspace $\mathcal{C}_+^\epsilon := \langle e_1, \ldots, e_k\rangle$, spanned by the rows of $G_+^\epsilon := (I_k \mid 0_{k \times (n-k)})$, is nondegenerate with inner discriminant $\delta_+ = 1$.
\item[(ii)] Define $v \in \FF_q^n$ by
\begin{itemize}
\item $v := \alpha\, e_k + \beta\, e_{k+1}$ if $k \leq n-2$;
\item $v := \alpha\, e_{n-1} + \beta\, e_n$ if $k = n-1$ and $\epsilon = +1$;
\item $v := e_n$ if $k = n-1$ and $\epsilon = -1$.
\end{itemize}
The subspace $\mathcal{C}_-^\epsilon := \langle e_1, \ldots, e_{k-1}, v\rangle$, spanned by the rows of the matrix $G_-^\epsilon$ with rows $e_1, \ldots, e_{k-1}, v$, is nondegenerate with inner discriminant $\delta_- = \gamma$.
\end{enumerate}
\end{lemma}
\begin{proof}
For (i), since $k \leq n-1$, none of the rows of $G_+^\epsilon$ touches the last column of $M^{(\epsilon)}$, so $G_+^\epsilon M^{(\epsilon)} (G_+^\epsilon)^T = I_k$, which has determinant $1 \neq 0$. Hence $\mathcal{C}_+^\epsilon$ is nondegenerate with $\delta_+ = 1$.

For (ii), $v$ is orthogonal (under $M^{(\epsilon)}$) to $e_1, \ldots, e_{k-1}$ by construction. The value $v M^{(\epsilon)} v^T$ equals $\alpha^2 + \beta^2 = \gamma$ in the first two cases, and $\tau_{-1} = \gamma$ in the third. Hence $G_-^\epsilon M^{(\epsilon)} (G_-^\epsilon)^T = \mathrm{diag}(1, \ldots, 1, \gamma)$, which has determinant $\gamma \neq 0$, so $\mathcal{C}_-^\epsilon$ is nondegenerate with $\delta_- = \gamma$.
\end{proof}

For $\delta \in \FF_q^*$, define the orthogonal group
\[
\mathbb{O}_m^\delta(q) := \{Q \in \mathrm{GL}_m(\FF_q) : Q \mathrm{diag}(1, \ldots, 1, \delta) Q^T = \mathrm{diag}(1, \ldots, 1, \delta)\}.
\]
The parameter $\delta$ serves as a generic discriminant placeholder; below it will be instantiated as $1$, $\gamma$, $\tau_\epsilon$, or $\tau_\epsilon\gamma$ depending on context; note that $\chi(\tau_\epsilon) = \epsilon$. In particular, $\mathbb{O}_n^{\tau_\epsilon}(q)$ is the orthogonal group preserving $M^{(\epsilon)}$, and its transitivity on each $\mathcal{N}^\epsilon_\pm[n, k]_q$ follows directly from Witt's extension theorem.

We use Witt's extension theorem; see~\cite{taylor1992geometry}.

\begin{theorem}[Witt's extension theorem]\label{thm:witt_extension}
Let $q$ be odd and let $B$ be a nondegenerate symmetric bilinear form on $V = \FF_q^n$. Any isometry $\phi: W_1 \to W_2$ between subspaces $W_1, W_2 \subseteq V$ extends to an isometry $\Phi: V \to V$ lying in the orthogonal group $\mathbb{O}(V, B)$.
\end{theorem}

\begin{lemma}[Transitivity via Witt extension]\label{lem:witt_transitivity}
For $\epsilon \in \{+1, -1\}$ and a fixed inner discriminant class $\delta \in \{1, \gamma\}$, the group $\mathbb{O}_n^{\tau_\epsilon}(q)$ acts transitively on the set of nondegenerate subspaces of dimension $k$ with inner discriminant $\delta$.
\end{lemma}
\begin{proof}
Let $W_1, W_2$ be two such subspaces. The restrictions $M^{(\epsilon)}|_{W_1}$ and $M^{(\epsilon)}|_{W_2}$ are nondegenerate symmetric bilinear forms on $k$-dimensional spaces, of common discriminant class $\delta$. By the classification of nondegenerate symmetric bilinear forms over $\FF_q$ ($q$ odd, Section~\ref{sec:preliminaries}), two such forms on a fixed-dimensional space are isometric if and only if they have equal discriminant class; hence there exists a linear isomorphism $\phi: W_1 \to W_2$ with $\phi(x) M^{(\epsilon)} \phi(y)^T = x M^{(\epsilon)} y^T$ for all $x, y \in W_1$. By Witt's extension theorem (Theorem~\ref{thm:witt_extension}) applied to $(\FF_q^n, M^{(\epsilon)})$, $\phi$ extends to an element $\Phi \in \mathbb{O}_n^{\tau_\epsilon}(q)$, and $\Phi(W_1) = W_2$.
\end{proof}

By Lemma~\ref{lem:witt_transitivity} applied to the reference subspaces $\mathcal{C}_\pm^\epsilon$ of Lemma~\ref{lem:reference_codes}, $\mathcal{N}^\epsilon_\pm[n, k]_q = \mathcal{C}_\pm^\epsilon \cdot \mathbb{O}_n^{\tau_\epsilon}(q)$, so by orbit--stabilizer,
\begin{equation}\label{eq:orbit_stab}
|\mathcal{N}^\epsilon_\pm[n, k]_q| = \frac{|\mathbb{O}_n^{\tau_\epsilon}(q)|}{|\mathrm{St}(\mathcal{C}_\pm^\epsilon)|}.
\end{equation}

\subsection{Stabilisers and Orthogonal Group Orders}

By~\cite[Corollary~30]{carlet2019new}, the stabiliser of $\mathcal{C}_\pm^\epsilon$ is block-diagonal in the basis adapted to $\mathcal{C}_\pm^\epsilon \perp_{M^{(\epsilon)}} (\mathcal{C}_\pm^\epsilon)^{\perp_{M^{(\epsilon)}}}$:
\begin{align*}
|\mathrm{St}(\mathcal{C}_+^\epsilon)| &= |\mathbb{O}_k^1(q)| \cdot |\mathbb{O}_{n-k}^{\tau_\epsilon}(q)|, \\
|\mathrm{St}(\mathcal{C}_-^\epsilon)| &= |\mathbb{O}_k^\gamma(q)| \cdot |\mathbb{O}_{n-k}^{\tau_\epsilon\gamma}(q)|,
\end{align*}
where the bottom-block discriminants come from the multiplicativity $\delta_W \cdot \delta_{W^\perp} \equiv \tau_\epsilon \pmod{(\FF_q^*)^2}$. The order of $\mathbb{O}_m^\delta(q)$ is given by~\cite[Eqs.~(11)--(12)]{carlet2019new} (citing Weyl~\cite{weyl1939classical}): for $n$ odd,
\[
|\mathbb{O}_n^\delta(q)| = 2 q^{(n-1)^2/4} \prod_{i=1}^{(n-1)/2}(q^{2i} - 1) \quad (\text{independent of } \delta),
\]
and for $n$ even,
\[
|\mathbb{O}_n^\delta(q)| = 2 q^{n(n-2)/4} \bigl(q^{n/2} - \chi((-1)^{n/2}\delta)\bigr) \prod_{i=1}^{n/2 - 1}(q^{2i} - 1).
\]

\subsection{Mass Formula}

Substituting orders into~\eqref{eq:orbit_stab},
\begin{equation}\label{eq:LB_pm}
\begin{aligned}
|\mathcal{N}^\epsilon_+[n,k]_q| &= \frac{|\mathbb{O}_n^{\tau_\epsilon}(q)|}{|\mathbb{O}_k^1(q)| \cdot |\mathbb{O}_{n-k}^{\tau_\epsilon}(q)|}, \\
|\mathcal{N}^\epsilon_-[n,k]_q| &= \frac{|\mathbb{O}_n^{\tau_\epsilon}(q)|}{|\mathbb{O}_k^\gamma(q)| \cdot |\mathbb{O}_{n-k}^{\tau_\epsilon\gamma}(q)|}.
\end{aligned}
\end{equation}

\begin{theorem}[Mass formula for $L^\epsilon$]\label{thm:Lepsilon}
For $\epsilon \in \{+1, -1\}$, $q$ odd, and $1 \leq k < n$,
\[
L^\epsilon(n, k, q) = \begin{cases}
q^{\frac{k(n-k)-1}{2}}\bigl(q^{n/2} - \epsilon\,\chi((-1)^{n/2})\bigr)\left[\begin{smallmatrix} \frac{n}{2}-1 \\ \frac{k-1}{2} \end{smallmatrix}\right]_{q^2}, & k \text{ odd, } n \text{ even,} \\[6pt]
L(n, k, q), & \text{otherwise.}
\end{cases}
\]
In particular, $L^{+1}(n, k, q) = L(n, k, q)$ recovers Carlet's formula~\cite[Corollary~32]{carlet2019new}.
\end{theorem}

\begin{proof}
We compute $L^\epsilon = |\mathcal{N}^\epsilon_+| + |\mathcal{N}^\epsilon_-|$ from~\eqref{eq:LB_pm} and compare with $L = L^{+1}$, the case treated by Carlet et al.~\cite[Corollary~32]{carlet2019new}.

When $m$ is even, write $|\mathbb{O}_m^\delta| = R_m \cdot A_m^\delta$, where $R_m := 2 q^{m(m-2)/4} \prod_{i=1}^{m/2 - 1}(q^{2i} - 1)$ is $\delta$-independent and $A_m^\delta := q^{m/2} - \chi((-1)^{m/2}\delta)$. Using $\chi(\gamma) = -1$ and $\chi(\tau_\epsilon) = \epsilon$, the pair $\{A_m^{\tau_\epsilon}, A_m^{\tau_\epsilon\gamma}\}$ satisfies
\[
A_m^{\tau_\epsilon} + A_m^{\tau_\epsilon\gamma} = 2 q^{m/2},
\qquad
A_m^{\tau_\epsilon} \cdot A_m^{\tau_\epsilon\gamma} = q^m - 1,
\]
both $\epsilon$-independent. When $m$ is odd, $|\mathbb{O}_m^\delta|$ does not depend on $\delta$. We treat the four parity cases.

\emph{Case $(k\text{ odd}, n\text{ odd})$.} Both $|\mathbb{O}_n^\delta|$ and $|\mathbb{O}_k^\delta|$ are $\delta$-independent, so
\[
L^\epsilon
= \frac{|\mathbb{O}_n|}{|\mathbb{O}_k|}\Bigl[\frac{1}{|\mathbb{O}_{n-k}^{\tau_\epsilon}|} + \frac{1}{|\mathbb{O}_{n-k}^{\tau_\epsilon\gamma}|}\Bigr]
= \frac{|\mathbb{O}_n|}{|\mathbb{O}_k|} \cdot \frac{A_{n-k}^{\tau_\epsilon} + A_{n-k}^{\tau_\epsilon\gamma}}{R_{n-k}\, A_{n-k}^{\tau_\epsilon} A_{n-k}^{\tau_\epsilon\gamma}},
\]
which is $\epsilon$-independent by the relations above.

\emph{Case $(k\text{ even}, n\text{ odd})$.} $|\mathbb{O}_n^\delta|$ and $|\mathbb{O}_{n-k}^\delta|$ are $\delta$-independent; the only $\delta$-dependence is in $|\mathbb{O}_k^\delta|$, which appears as $|\mathbb{O}_k^1|$ in $|\mathcal{N}^\epsilon_+|$ and $|\mathbb{O}_k^\gamma|$ in $|\mathcal{N}^\epsilon_-|$ irrespective of $\epsilon$. Hence $L^\epsilon$ has no $\epsilon$-dependence.

\emph{Case $(k\text{ even}, n\text{ even})$.} All three orders depend on $\delta$. Set $X = q^{k/2}$, $Y = q^{(n-k)/2}$, $s = \chi((-1)^{k/2})$, $t = \chi((-1)^{(n-k)/2})$, so that $A_n^{\tau_\epsilon} = XY - \epsilon\, st$, $A_k^\delta = X - \chi(\delta) s$, $A_{n-k}^\delta = Y - \chi(\delta) t$. Direct computation gives
\[
\begin{aligned}
L^\epsilon
&= \frac{R_n A_n^{\tau_\epsilon}}{R_k R_{n-k}}\Bigl[\frac{1}{A_k^1 A_{n-k}^{\tau_\epsilon}} + \frac{1}{A_k^\gamma A_{n-k}^{\tau_\epsilon\gamma}}\Bigr]
= \frac{R_n A_n^{\tau_\epsilon}}{R_k R_{n-k}} \cdot \frac{2(XY + \epsilon\, st)}{(X^2 - 1)(Y^2 - 1)} \\
&= \frac{R_n}{R_k R_{n-k}} \cdot \frac{2((XY)^2 - 1)}{(X^2 - 1)(Y^2 - 1)},
\end{aligned}
\]
using $\epsilon^2 = s^2 = t^2 = 1$. Hence $L^\epsilon$ is $\epsilon$-independent.

\emph{Case $(k\text{ odd}, n\text{ even})$.} $|\mathbb{O}_k^\delta|$ and $|\mathbb{O}_{n-k}^\delta|$ are both $\delta$-independent, so $|\mathcal{N}^\epsilon_+| = |\mathcal{N}^\epsilon_-|$ and
\[
L^\epsilon(n, k, q) = \frac{2|\mathbb{O}_n^{\tau_\epsilon}|}{|\mathbb{O}_k| \cdot |\mathbb{O}_{n-k}|}.
\]
The ratio $L^\epsilon/L = A_n^{\tau_\epsilon}/A_n^{1} = (q^{n/2} - \epsilon\,\chi((-1)^{n/2}))/(q^{n/2} - \chi((-1)^{n/2}))$ substituted into Carlet's formula (Appendix~\ref{app:mass_odd}) gives the stated expression.
\end{proof}

\section{Validation Script}
\label{app:script}

The following SageMath script implements Algorithm~\ref{alg:gi_reduction} and compares it, for the parameters listed in Section~\ref{sec:algorithm}, with a brute-force permutation-equivalence test on all pairs of codes of equal dimension (the GI step is performed by exhaustive search over $S_n$ through a canonical form of the projector).

\begin{lstlisting}[style=sage]
# Validate Algorithm 1 against brute-force permutation equivalence
# on ALL pairs of codes, using per-subspace canonical forms.
import itertools, sys, time

def run(q, n):
    t0 = time.time()
    F = GF(q); V = VectorSpace(F, n); one = vector(F, [1]*n)
    perms = [Permutation(p).to_matrix().change_ring(F)
             for p in Permutations(n)]
    I = identity_matrix(F, n); J = matrix(F, n, n, lambda i, j: 1)
    MJ = None
    if q % 2 == 1:
        b = next(b for b in F if b != 0 and 1 + n*b != 0)
        MJ = I + b*J
    def key(M):
        return tuple(tuple(r) for r in M.rows())
    total = bad = notgi = 0
    for k in range(1, n):
        info = []
        for C in V.subspaces(k):
            G = C.basis_matrix(); A = G*G.T; h = k - A.rank()
            s = None
            if h == 1:
                x = A.left_kernel().basis()[0]*G; s = (x*one != 0)
            # canonical form of the code (brute force over S_n)
            ccan = min(key(V.subspace([v*P for v in C.basis()])
                           .basis_matrix()) for P in perms)
            pcan = None
            if h == 0 or (h == 1 and s):
                M = I if h == 0 else MJ
                AM = G*M*G.T; assert AM.is_invertible()
                Pi = M*G.T*AM.inverse()*G
                # canonical form of the projector (brute-force GI)
                pcan = min(key(P.T*Pi*P) for P in perms)
            info.append((h, s, ccan, pcan))
        pairs = itertools.combinations_with_replacement(info, 2)
        for (h1, s1, c1, p1), (h2, s2, c2, p2) in pairs:
            total += 1
            if h1 != h2: out = "neq"
            elif h1 >= 2: out = "notGI"
            elif h1 == 0: out = "eq" if p1 == p2 else "neq"
            else:
                if s1 != s2: out = "neq"
                elif not s1: out = "notGI"
                else: out = "eq" if p1 == p2 else "neq"
            truth = "eq" if c1 == c2 else "neq"
            if out == "notGI":
                notgi += 1
                gi1 = h1 <= 1 and (h1 == 0 or s1)
                gi2 = h2 <= 1 and (h2 == 0 or s2)
                if gi1 or gi2: bad += 1
            elif out != truth: bad += 1
    print(f"q={q} n={n} ALL pairs: pairs={total} notGI={notgi} "
          f"mismatches={bad} ({time.time()-t0:.0f}s)", flush=True)

for q, n in [(2,3),(3,3),(5,3),(2,4),(3,4),(5,4),(2,5),(3,5)]:
    run(q, n)
\end{lstlisting}

\bibliographystyle{plain}
\bibliography{main}

\end{document}